\documentclass[12pt,draftclsnofoot,onecolumn]{IEEEtran}
\usepackage{ifpdf}
\usepackage{cite}
\usepackage[cmex10]{amsmath}
\usepackage{array}
\usepackage{amsfonts}
\usepackage{mathrsfs}
\usepackage{arydshln}
\usepackage{slashbox}
\usepackage{graphicx}
\usepackage{float}
\usepackage{subfigure}
\usepackage{amssymb}
\usepackage{amsmath}
\usepackage{hyperref}
\usepackage{multirow}
\usepackage{multicol}
\usepackage{cite}
\usepackage[subfigure]{graphfig}
\usepackage{xcolor}
\usepackage{setspace}
\usepackage[ruled,boxed,linesnumbered]{algorithm2e}
\usepackage{url}
\newcommand{\mv}[1]{\mbox{\boldmath{$ #1 $}}}

\newtheorem{lemma}{\underline{Lemma}}
\newtheorem{proposition}{\underline{Proposition}}

\newtheorem{remark}{\underline{Remark}}
\newenvironment{proof}[1][Proof]{\begin{trivlist}
\item[\hskip \labelsep {\mvseries #1}]}{\end{trivlist}}

\newcommand{\qed}{\nobreak \ifvmode \relax \else
      \ifdim\lastskip<1.5em \hskip-\lastskip
      \hskip1.5em plus0em minus0.5em \fi \nobreak
      \vrule height0.75em width0.5em depth0.25em\fi}
\begin{document}
\title{{Constant Envelope Precoding for MIMO Systems}
\footnote{This work was presented in part at the IEEE International Conference on Communications (ICC), Kuala Lumpur, Malaysia, May 23-27, 2016 \cite{ICC16}.}\footnote{S. Zhang is with the NUS Graduate School for Integrative Sciences and Engineering (NGS), National University of Singapore (e-mail:shuowen.zhang@u.nus.edu). She is also with the Department of Electrical and Computer Engineering, National University of Singapore.}\footnote{R. Zhang is with the
Department of Electrical and Computer Engineering, National
University of Singapore (e-mail:elezhang@nus.edu.sg). He is also
with the Institute for Infocomm Research, A*STAR, Singapore.}\footnote{T. J. Lim is with the Department of Electrical
and Computer Engineering, National University of Singapore
(e-mail:eleltj@nus.edu.sg).}}
\author{\IEEEauthorblockN{Shuowen~Zhang, Rui~Zhang, and Teng Joon Lim}}
\maketitle
\vspace{-12.9mm}
\begin{abstract}
Constant envelope (CE) precoding is an appealing transmission technique, which enables highly efficient power amplification, and is realizable with a single radio frequency (RF) chain at the multi-antenna transmitter. In this paper, we study the transceiver design for a point-to-point multiple-input multiple-output (MIMO) system with CE precoding. Both single-stream transmission (i.e., beamforming) and multi-stream transmission (i.e., spatial multiplexing) are considered. For single-stream transmission, we optimize the receive beamforming vector to minimize the symbol error rate (SER) for any given channel realization and desired constellation at the combiner output. By reformulating the problem as an equivalent quadratically constrained quadratic program (QCQP), we propose an efficient semi-definite relaxation (SDR) based algorithm to find an approximate solution. Next, for multi-stream transmission, we propose a new scheme based on antenna grouping at the transmitter and minimum mean squared error (MMSE) or zero-forcing (ZF) based beamforming at the receiver. The transmit antenna grouping and receive beamforming vectors are then jointly designed to minimize the maximum SER over all data streams. Finally, the error-rate performance of single- versus multi-stream transmission is compared via simulations under different setups.
\end{abstract}
\vspace{-3.9mm}
\section{Introduction}
Motivated by the demand for power-efficient and cost-effective radio frequency (RF) components in wireless communication systems, there has been an upsurge of research interests in \emph{constant envelope (CE) precoding} for multi-antenna or multiple-input multiple-output (MIMO) communications \cite{SUCE,JSTSPCE,CEadaptive,CEMulticast,MUCE,improved,freqCE,phaseconstraint}. Specifically, under the so-called per-antenna CE constraint that restricts the equivalent complex baseband signal at each transmit antenna to have constant amplitude, CE precoding performs a mapping (which is generally nonlinear) from the desired information-bearing symbols to solely the transmitted signal phases at multiple antennas, based on the instantaneous channel state information (CSI). In practice, the transmitted signal phases can be controlled at either the baseband or the RF band, which correspond to two transmitter architectures for realizing CE precoding as shown in Fig. \ref{CEdiagram} (a) and (b), respectively.
\begin{figure}[ht]
  \centering
  \subfigure[Architecture I: Baseband phase control]{
    \label{Multi_RF}
    \includegraphics[width=5.1in]{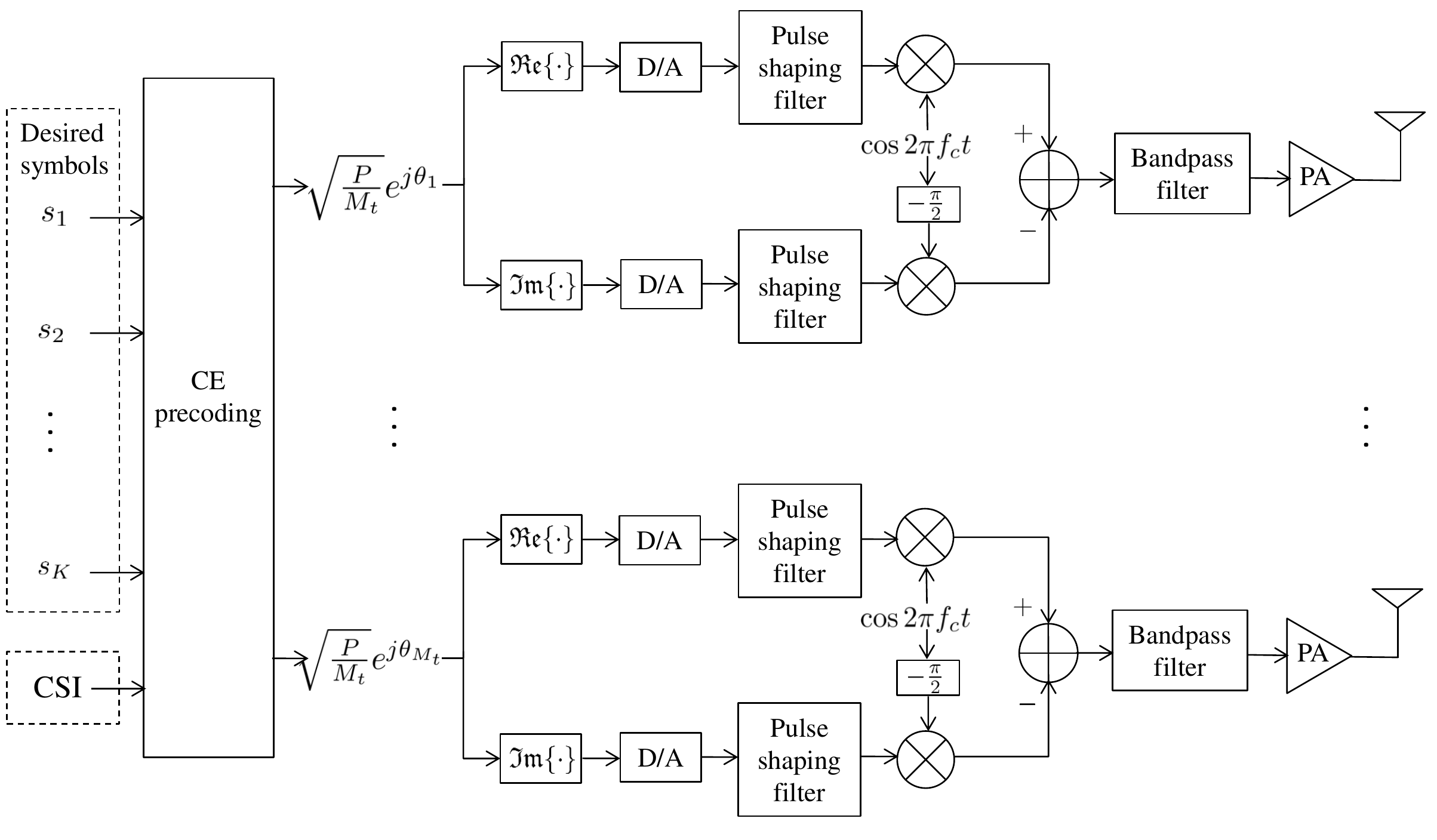}}
    \hspace{0in}
  \subfigure[Architecture II: RF-band phase control]{
    \label{Single_RF}
    \includegraphics[width=3.5in]{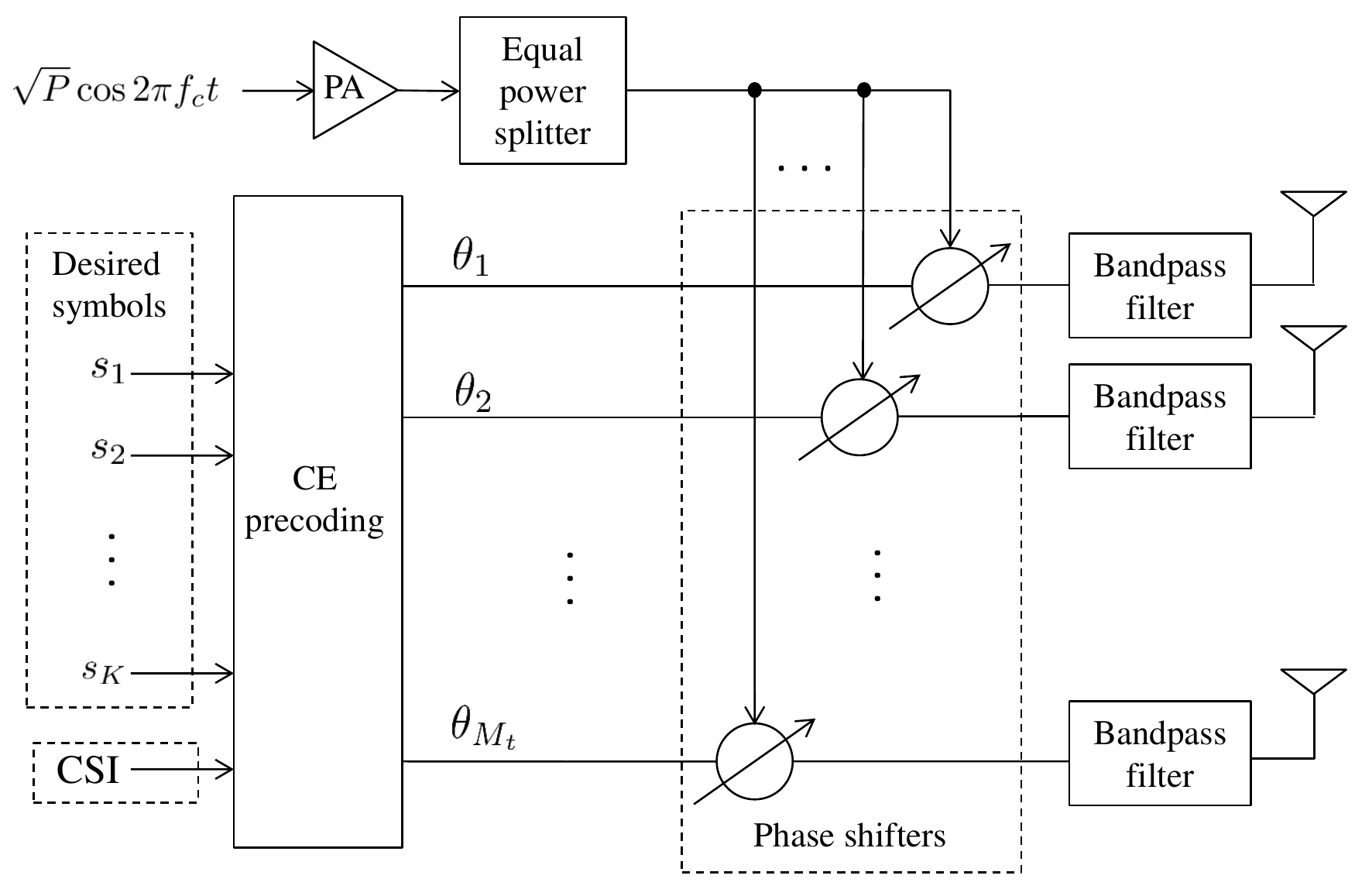}}
  \caption{Two transmitter architectures for CE precoding in MIMO system.}
  \label{CEdiagram}
\end{figure}

CE precoding is advantageous compared to its non-CE counterparts due to the following reasons. First, notice from Fig. \ref{CEdiagram} that for each of the two transmitter architectures, the input RF signal for the power amplifier (PA) is a continuous-time CE signal, which in general leads to high PA efficiency.\footnote{It is worth noting that for Architecture I, the RF signal at each antenna is generally not a perfect CE signal in practice, due to the use of non-ideal (e.g., raised-cosine) pulse shaping filters, which may compromise the effectively harvested PA efficiency gain. However, compared to other non-CE precoded signals, such quasi-CE signals can still achieve better PA efficiency in general.} Since highly efficient PAs (e.g., class-C and switched mode PAs) in practice have nonlinear amplitude transfer functions, they can only be used with CE input signals or else output distortion arises \cite{RFPA}. Moreover, since CE signals have the lowest possible peak-to-average power ratio (PAPR), they require the minimum backoff for operation when linear PAs (e.g., class-A and class-B PAs) are used, thus achieving high efficiency \cite{RFPA}. The low PAPR of CE signals also allows for the use of less expensive PAs with smaller dynamic range.
Furthermore, for conventional non-CE precoding techniques, both the amplitude and phase of the equivalent complex baseband signal at each transmit antenna need to vary depending on the instantaneous channel and/or symbol realization. Therefore, they are practically performed in the digital domain and require a dedicated RF chain for each transmit antenna, which is both costly and power-consuming. In contrast, although CE precoding can be similarly realized using Architecture I in Fig. \ref{CEdiagram} (a), as proposed by prior work (e.g., \cite{SUCE}), it can be alternatively implemented directly in the RF domain by using a network of digitally-controlled phase shifters with a single RF chain as shown in Architecture II in Fig. 1 (b), which was first proposed in \cite{CEMulticast}.\footnote{It is worth noting that with Architecture II, the transmission rate is limited by the switching speed of the phase shifters, which, however, can be very fast in practice (down to the scale of nanoseconds per use \cite{phaseshifter}).} Note that despite the single RF chain, this architecture can even support concurrent transmission of multiple data streams (i.e., spatial multiplexing) via CE precoding.\footnote{Note that there is another line of work on the so-called load modulated MIMO (see e.g., \cite{WSA14}) which also supports spatial multiplexing with a single RF chain. Different from Architecture II, this scheme requires an additional matching network and multiple real-time adjustable load modulators at the transmitter.}

However, the aforementioned benefits of CE precoding come at the cost of a transceiver design that is more challenging than those with the conventional average-based sum power constraint (SPC) and per-antenna power constraint (PAPC) (see, e.g., \cite{linearprecoding,dualityPAPC,LTCC,multicellBD,massivePAPC}), which are less restrictive. In \cite{SUCE,JSTSPCE,CEadaptive}, a single-user multiple-input single-output (MISO) system with the per-antenna CE constraint is studied. It is shown in \cite{SUCE,JSTSPCE} that by varying the transmitted signal phases, the noise-free signal at the receiver always lies in an annular region, whose boundaries are characterized by the instantaneous channel realization and per-antenna transmit power. Moreover, efficient CE precoding algorithms are proposed in \cite{SUCE,JSTSPCE} to find the nonlinear mapping from any desired received signal point within the annulus to the transmitted signal phases based on the instantaneous CSI. Furthermore, note that a desired receiver constellation is feasible for CE precoding in a MISO channel if and only if it can be scaled to lie in the annulus, such that the corresponding transmitted CE signals can be found for all the signal points in the desired constellation. Therefore, for a fading channel that yields a time-varying annulus, a fixed receiver constellation may not be always feasible, thus resulting in severe reliability degradation (assuming transmission is blocked when the constellation is not feasible). To resolve this issue, both fixed-rate and variable-rate adaptive receiver constellation designs are proposed in \cite{CEadaptive} for CE precoding in MISO fading channel.

In addition, CE precoding has been investigated in various multi-user systems. For MISO multicast systems with the common symbol for multiple users drawn from a given constellation, the joint optimization of transmitter CE precoding and receiver constellation scaling and rotation for symbol error rate (SER) minimization is studied in \cite{CEMulticast}. For large-scale MISO broadcast systems, low-complexity CE precoding algorithms are proposed in \cite{MUCE,improved} for frequency-flat channels and in \cite{freqCE} for frequency-selective channels. It is shown that with a sufficiently large number of transmit antennas, arbitrarily low multi-user interference (MUI) power can be achieved at each user with the proposed schemes in \cite{MUCE,improved,freqCE}. As an extension to \cite{freqCE}, an efficient CE precoding scheme is proposed in \cite{phaseconstraint} considering an additional constraint on the signal phase variation at each transmit antenna between consecutive channel uses, such that the spectral regrowth resulting from abrupt phase changes can be potentially eliminated. It is shown that the extra transmit power required for the proposed scheme in \cite{phaseconstraint} to achieve the same transmission rate as that in \cite{freqCE} is small.

In this paper, we study the transceiver design in a point-to-point MIMO system with CE precoding, assuming perfect CSI is available at both the transmitter and the receiver.\footnote{How to efficiently obtain the MIMO channel knowledge at the transmitter using either of the two architectures in Fig. \ref{CEdiagram} is an interesting problem, which is left for our future work.} Both single-stream transmission (i.e., beamforming) and multi-stream transmission (i.e., spatial multiplexing) are considered. Our main contributions are summarized as follows:
\begin{itemize}
\item For single-stream transmission, we consider the problem of receive beamforming vector optimization to minimize the SER at the combiner (beamforming) output, for any given channel realization and desired constellation at the combiner output. Specifically, by approximating the exact SER with its union bound, we formulate the equivalent problem of maximizing the minimum Euclidean distance (MED) between any two signal points at the combiner output while guaranteeing the feasibility of the constellation (i.e., it can be scaled to lie in an annular region characterized by the channel realization and receive beamforming). We first show that this problem is feasible for any desired constellation at the combiner output if the rank of the channel matrix is no smaller than two, which always holds under our assumed independent and identically distributed (i.i.d.) Rayleigh fading MIMO channel. Then, we introduce an auxiliary vector to reformulate this problem into an equivalent quadratically constrained quadratic program (QCQP). By applying the semi-definite relaxation (SDR) technique as well as our customized Gaussian randomization methods, we propose an efficient algorithm to find an approximate solution to the QCQP.
\item Next, for multi-stream transmission, CE precoding that maps the symbols of multiple data streams to the transmitted signal phase at each antenna generally needs to be jointly designed with the MIMO receiver, which is a complicated problem to solve in general. To tackle this problem, we propose a new scheme based on transmit antenna grouping and minimum mean squared error (MMSE) or zero-forcing (ZF) based receive beamforming, which decouples the joint transceiver design problem for the multiple data streams to a set of parallel sub-problems, one for each data stream. The transmit antenna grouping and receive beamforming vectors are then jointly optimized to minimize the maximum SER over all data streams, subject to the constellation feasibility constraints. Finally, the error-rate performance of single-stream and multi-stream transmissions is compared under various practical setups.
\end{itemize}

The remainder of this paper is organized as follows. Section II introduces the system model for CE precoding for the cases of single-stream transmission and multi-stream transmission, respectively. Section III presents the receiver optimization problem for single-stream transmission and proposes an efficient solution. Section IV presents our scheme for the MIMO transceiver design for multi-stream transmission. Numerical results are provided in Section V to evaluate the performance of the proposed schemes. Finally, Section VI concludes the paper.

\textit{Notations}: Scalars and vectors are denoted by lower-case letters and boldface lower-case letters, respectively. $|z|$, $z^*$, $\arg\{z\}$ and $\mathfrak{Re}\{z\}$ denote the absolute value, the conjugate, the angle and the real part of a complex scalar $z$, respectively. $\|{\mv{z}}\|_p$ and $z_k$ denote the $l_p$-norm and the $k$th element of a vector ${\mv{z}}$, respectively. $\mathbb{C}^{M\times N}$ denotes the space of $M\times N$ complex matrices. ${\mv{I}}_M$ denotes the $M\times M$ identity matrix, and ${\mv{0}}$ denotes an all-zero matrix with appropriate dimension. For an $M\times N$ matrix ${\mv{A}}$, ${\mv{A}}^T$ and ${\mv{A}}^H$ denote its transpose and conjugate transpose, respectively; $\mathrm{rank}({\mv{A}})$ and $[{\mv{A}}]_{i,j}$ denote the rank of ${\mv{A}}$ and the $(i,j)$-th element of ${\mv{A}}$, respectively. The null space of ${\mv{A}}$ is defined as $\mathrm{Null}({\mv{A}})\overset{\Delta}{=}\{{\mv{x}}\in \mathbb{C}^{N\times 1}:{\mv{Ax}}={\mv{0}}\}$. For a square matrix $\mv{S}$, $\mathrm{tr}({\mv{S}})$ denotes its trace, and ${\mv{S}}\succeq {\mv{0}}$ means that ${\mv{S}}$ is positive semi-definite. The distribution of a circularly symmetric complex Gaussian (CSCG) random variable with mean $\mu$ and variance $\sigma^2$ is denoted by $\mathcal{CN}(\mu,\sigma^2)$; and $\sim$ stands for ``distributed as''. $\max\{x,y\}$ and $\min\{x,y\}$ denote the maximum and the minimum of two real numbers $x$ and $y$, respectively. $\mathbb{E}[\cdot]$ denotes the expectation operator.
\section{System Model}
Consider a point-to-point MIMO system with $M_t\geq 2$ antennas at the transmitter and $M_r\geq 2$ antennas at the receiver. We assume a quasi-static flat-fading environment with $\tilde{\mv{H}}\in \mathbb{C}^{M_r\times M_t}$ denoting the equivalent complex baseband channel matrix. For convenience, the entries of $\tilde{\mv{H}}$ are modeled by i.i.d. CSCG random variables with equal variance of $\beta$, i.e., $[\tilde{\mv{H}}]_{i,j}\sim \mathcal{CN}(0,\beta),\ \forall i,\ \forall j$, where $\beta$ specifies the average channel power attenuation due to path loss and shadowing; while our proposed design is applicable to arbitrary channel realizations. Note that under the above assumption, we have $\mathrm{rank}(\tilde{\mv{H}})=\min\{M_r,M_t\}$, i.e., $\tilde{\mv{H}}$ is a full-rank matrix, with probability one. For both transmitter architectures in Fig. \ref{CEdiagram}, the baseband transmission is modeled by
\begin{equation}
\tilde{{\mv{y}}}=\tilde{\mv{H}}{\mv{x}}+\tilde{\mv{n}},\label{channel}
\end{equation}
where $\tilde{\mv{y}}\in \mathbb{C}^{M_r\times1}$ and ${\mv{x}}\in \mathbb{C}^{M_t\times1}$ denote the received and the transmitted signal vectors, respectively; $\tilde{\mv{n}}\sim \mathcal{CN}(\mv{0},\sigma^2{\mv{I}}_{M_r})$ denotes the $M_r\times1$ CSCG noise vector at the receiver. We consider CE precoding at the transmitter, under the assumption that $\tilde{\mv{H}}$ is perfectly known at both the transmitter and the receiver. As shown in Fig. \ref{CEdiagram}, we assume a total transmit power denoted by $P$, which is equally allocated to the $M_t$ transmit antennas. With CE precoding, the equivalent complex baseband signal at each transmit antenna is expressed as
\begin{equation}
x_i=\sqrt{\frac{P}{M_t}}e^{j\theta_i},\quad i=1,...,M_t,\label{CEconstraint}
\end{equation}
where information is modulated in the transmitted signal phases $\theta_i\in[0,2\pi),\ i=1,...,M_t$.

Let $\bar{R}$ denote the transmission rate in bits/second/hertz (bps/Hz). We assume that the transmitted bit sequence is demultiplexed into $K,\ K\leq \min\{M_r,M_t\}$ data streams, each carrying $\frac{\bar{R}}{K}$ bits. For convenience of modulation, we assume $\frac{\bar{R}}{K}$ is an integer. Each data stream is further assumed to be modulated with the same constellation denoted by $\mathcal{S}$, which is of size $N=2^{\frac{\bar{R}}{K}}$. In the following, we present the transceiver model of the above system for the cases of single-stream transmission (i.e., $K=1$) and multi-stream transmission (i.e., $K\geq 2$), respectively.
\vspace{-1mm}
\subsection{Single-Stream Transmission}
For single-stream transmission with $K=1$, we let ${\mv{u}}\in \mathbb{C}^{M_r\times1}$ denote the receive beamforming vector, which is assumed to be normalized such that $\|{\mv{u}}\|_2=1$ without loss of generality. After applying the receive beamforming, the combiner output signal is given by
\vspace{-1mm}\begin{equation}
y={\mv{u}}^H\tilde{{\mv{y}}}={\mv{u}}^H\tilde{\mv{H}}{\mv{x}}+{n},\label{combineroutput}
\end{equation}
where ${\mv{u}}^H\tilde{\mv{H}}$ is the effective MISO channel from the transmitter to the combiner output, and ${n}={\mv{u}}^H\tilde{\mv{n}}$ denotes the effective noise, whose distribution can be shown to be given by ${n}\sim \mathcal{CN}(0,\sigma^2)$. Let $d\overset{\Delta}{=} {\mv{u}}^H\tilde{\mv{H}}{\mv{x}}= \sqrt{\frac{P}{M_t}}{\mv{u}}^H\tilde{\mv{H}}\left[e^{j\theta_1},...,e^{j\theta_{M_t}}\right]^T$ denote the noise-free signal at the combiner output.

First, note that the constellation $\mathcal{S}$ is feasible at the combiner output if and only if there exists a scaling factor $\alpha>0$, such that any symbol point on $\alpha \mathcal{S}$ can be mapped to CE signals at the transmitter, i.e., the following problem is feasible for any $s \in \mathcal{S}$:
\vspace{-1mm}\begin{align}
\mathrm{find}\quad &\{\theta_i:\theta_i\in [0,2\pi)\}_{i=1}^{M_t}\label{BF_S}\\
\mathrm{s.t.}\quad &d=\alpha s.\nonumber
\end{align}
By generalizing the results in \cite{SUCE,JSTSPCE} for the MISO channel, the feasible region of $d$ with a given $\mv{u}$ and $\theta_i\in [0,2\pi),\ \forall i$ can be shown to be given by
\vspace{-1mm}\begin{equation}
\mathcal{D}({\mv{u}})=\{d\in \mathbb{C}:r({\mv{u}})\leq |d|\leq R({\mv{u}}) \},\label{BFregion}
\end{equation}
where
\vspace{-1mm}\begin{align}
R({\mv{u}})=&\sqrt{\frac{P}{M_t}}\|{\mv{u}}^H\tilde{\mv{H}}\|_1,\label{R}\\ r({\mv{u}})=&\sqrt{\frac{P}{M_t}}\max\left\{2\|{\mv{u}}^H\tilde{\mv{H}}\|_{\infty}-\|{\mv{u}}^H\tilde{\mv{H}}\|_1,0\right\}.\label{r}
\end{align}
As a result of (\ref{BFregion}), $\mathcal{S}$ is feasible if and only if $\alpha>0$ exists such that $\alpha \mathcal{S}\subset \mathcal{D}({\mv{u}})$, or equivalently,
\vspace{-3mm}\begin{align}
\frac{r({\mv{u}})}{R({\mv{u}})}\leq\frac{\underset{s\in \mathcal{S}}{\min}|s|}{\underset{s\in \mathcal{S}}{\max}|s|}.\label{feascond}
\end{align}\vspace{-3mm}

Moreover, for any feasible $\mathcal{S}$ and the corresponding $\alpha$, efficient CE precoding algorithms proposed in \cite{SUCE,JSTSPCE} can be used to find the solution to Problem (\ref{BF_S}) for any $s\in \mathcal{S}$ based on ${\mv{u}}^H\tilde{\mv{H}}$, where the mapping from $\alpha s$ to ${\mv{x}}$ is generally nonlinear, in contrast to conventional linear precoding techniques. The details of these algorithms are omitted here for brevity. Therefore, the combiner output signal in (\ref{combineroutput}) is equivalently represented by
\vspace{-2mm}\begin{equation}
y=\alpha s+{n},\quad s\in {\cal S}.\label{equichannel}
\end{equation}\vspace{-2mm}
\begin{figure}[t]
  \centering
  \subfigure[Infeasible case with ${\mv{u}}^{(1)}={[1\ 0]}^T$]{
    \includegraphics[width=2.35in]{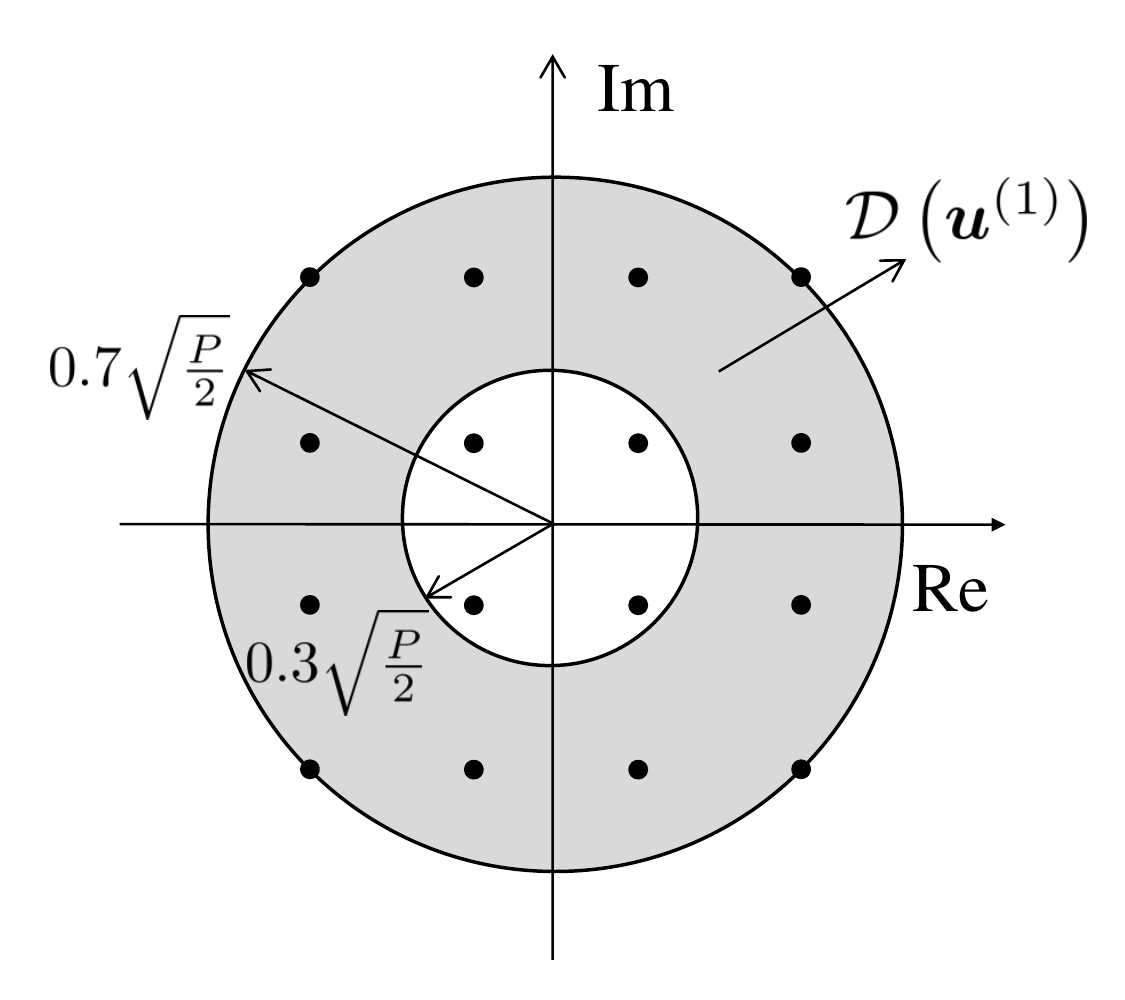}}
  \hspace{0in}
  \subfigure[Feasible case with ${\mv{u}}^{(2)}={[0\ 1]}^T$]{
    \includegraphics[width=2.35in]{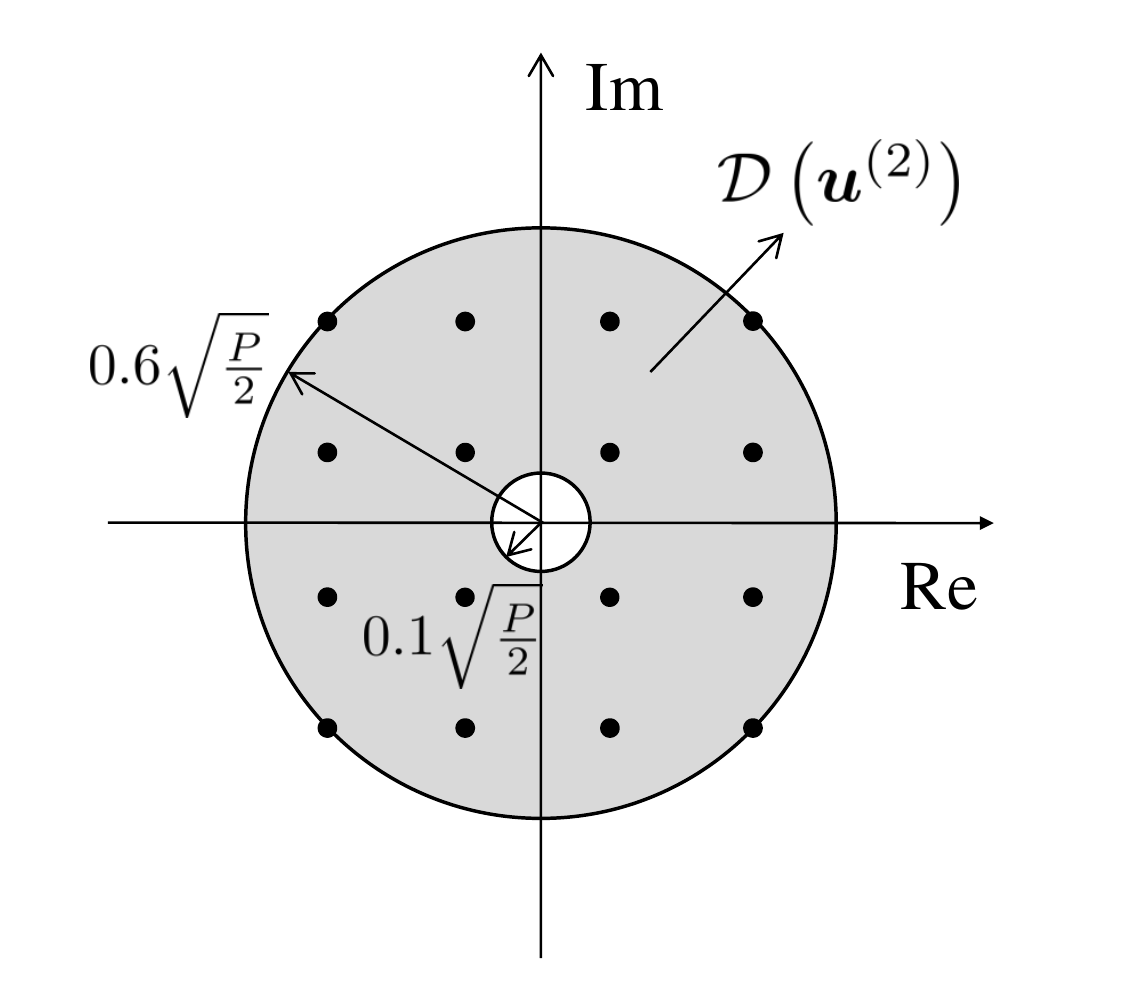}}
  \caption{Feasibility of 16-QAM for CE single-stream transmission with given $\tilde{\mv{H}}$ and different $\mv{u}$.}\label{BFfeasibility}
\end{figure}
\vspace{-8mm}
\begin{remark}
It is worth noting that in order to maximize the signal power at the combiner output and yet meet the feasibility constraint of $\alpha \mathcal{S}\subset \mathcal{D}({\mv{u}})$, we should set $\alpha=R({\mv{u}})=\sqrt{\frac{P}{M_t}}\|{\mv{u}}^H\tilde{\mv{H}}\|_1$ in (\ref{equichannel}) for any $\mathcal{S}$ that is feasible and satisfies $\underset{s\in\mathcal{S}}{\max} |s|=1$, such that the signal point with the largest amplitude in $\mathcal{S}$ lies on the outer boundary of $\mathcal{D}({\mv{u}})$ at the combiner output.
\end{remark}

Note that for given channel $\tilde{\mv{H}}$, $\mathcal{D}({\mv{u}})$ as well as the feasibility of $\mathcal{S}$ depends on the receive beamforming vector $\mv{u}$. For example, consider the case where $M_t=M_r=2$, $|[\tilde{\mv{H}}]_{1,1}|=0.5,\ |[\tilde{\mv{H}}]_{1,2}|=0.2,\ |[\tilde{\mv{H}}]_{2,1}|=0.35,\ |[\tilde{\mv{H}}]_{2,2}|=0.25$, and $\mathcal{S}$ is a 16-QAM (quadrature amplitude modulation) constellation (i.e., $\underset{s\in \mathcal{S}}{\min}|s|/\underset{s\in\mathcal{S}}{\max} |s|=\frac{1}{3}$). As shown in Fig. \ref{BFfeasibility}, $\mathcal{S}$ is infeasible with ${\mv{u}}^{(1)}=[1\ 0]^T$, but is feasible with ${\mv{u}}^{(2)}=[0 \ 1]^T$. As a result, we are motivated to investigate the design of $\mv{u}$ based on the channel realization $\tilde{\mv{H}}$ and desired constellation $\mathcal{S}$, which will be detailed in Section III.
\vspace{-1mm}
\subsection{Multi-Stream Transmission}
For multi-stream transmission with $K\geq 2$, we assume a linear receiver is used to decode $s_k$'s. Specifically, let ${\mv{u}}_k\in \mathbb{C}^{M_r\times 1}$ denote the receive beamforming vector for decoding $s_k$, which is assumed to be normalized such that $\|{\mv{u}}_k\|_2=1$ without loss of generality. Applying ${\mv{u}}_k^H$ to the received signal vector in (\ref{channel}) yields
\vspace{-1mm}\begin{align}
y_k={\mv{u}}_k^H\tilde{\mv{y}}={\mv{u}}_k^H\tilde{\mv{H}}{\mv{x}}+n_k,\label{SMchannel}
\end{align}
where ${\mv{u}}_k^H\tilde{\mv{H}}$ is the effective MISO channel from the transmitter to the combiner output of the $k$th data stream, and $n_k={\mv{u}}_k^H\tilde{\mv{n}}$, with $n_k\sim \mathcal{CN}(0,\sigma^2)$. Let $d_k\overset{\Delta}{=} {\mv{u}}_k^H\tilde{\mv{H}}{\mv{x}}= \sqrt{\frac{P}{M_t}}{\mv{u}}_k^H\tilde{\mv{H}}\left[e^{j\theta_1},...,e^{j\theta_{M_t}}\right]^T$ denote the noise-free signal received at the $k$th data stream. Note that $d_k$'s are coupled with all $\theta_i$'s, which introduces the following challenges to the transceiver design:
\begin{itemize}
\item First, note that with given $\{{\mv{u}}_k\}_{k=1}^K$, $\mathcal{S}$ is feasible for the $K$ data streams if and only if there exists a set of scaling factors $\{\alpha_k:\alpha_k>0\}_{k=1}^K$, such that the following problem is feasible for any $\{s_k:\ s_k\in \mathcal{S}\}_{k=1}^K$:
    \vspace{-1mm}\begin{align}
    \mathrm{find}\quad &\{\theta_i:\theta_i\in [0,2\pi)\}_{i=1}^{M_t}\label{SM_S}\\
    \mathrm{s.t.}\quad &d_k=\alpha_k s_k,\quad k=1,...,K.\nonumber
    \end{align}
However, this condition is in general difficult to verify when $K>1$. Specifically, it is hard to check the feasibility of Problem (\ref{SM_S}) for given $\{\alpha_k,s_k\}_{k=1}^K$, since the jointly feasible region for $\{d_k\}_{k=1}^K$ with $\theta_i\in [0,2\pi),\ \forall i$ is difficult to characterize.\footnote{Specifically, although the marginally feasible region of each $d_k$ can be shown to be still an annular region (same as the case of single-stream transmission), $\{d_k\}_{k=1}^K$ from all $K$ data streams may not be jointly feasible with each $d_k$ arbitrarily drawn from its corresponding annular region.}
\item Second, even assuming $\mathcal{S}$ is verified to be feasible with given $\{{\mv{u}}_k\}_{k=1}^K$, it is hard to find the mapping from desired $\{\alpha_k,s_k\}_{k=1}^K$ to the transmitted signal phases $\{\theta_i\}_{i=1}^{M_t}$ by solving Problem (\ref{SM_S}), which is a non-convex problem and is more difficult to solve than Problem (\ref{BF_S}) for the case of $K=1$.\footnote{Note that for a massive MIMO system with sufficiently large $M_t$, it can be shown that any desired constellations for the $K$ data streams are jointly feasible regardless of the channel realization, and Problem (\ref{SM_S}) can be solved via algorithms proposed in e.g., \cite{MUCE}. However, this is due to the excessive degrees of freedom available at the transmitter, and is in general not true for finite value of $M_t$.}
\item Third, note that both the feasibility of $\mathcal{S}$ and the CE precoding design depend on the receive beamforming vectors $\{{\mv{u}}_k\}_{k=1}^K$. However, due to the lack of effective methods to deal with the above problems, it is difficult to formulate a problem to optimize $\{{\mv{u}}_k\}_{k=1}^K$ directly.
\end{itemize}

To overcome these challenges, we propose a new scheme that decouples the CE precoding design for the $K$ data streams, by adopting \emph{antenna grouping} at the transmitter. Specifically, the transmit antennas are divided into $K$ groups with equal size $\frac{M_t}{K}$, each assigned to the transmission of one data stream. For the purpose of exposition, we assume $\frac{M_t}{K}$ is an integer in the sequel. Let $\tilde{\mv{H}}_k\in \mathbb{C}^{M_r\times\frac{M_t}{K}}$ denote the channel matrix from the transmit antennas in the $k$th group to the receiver. For convenience of illustration, we assume the grouping is based on antenna index, i.e., the first group consists of transmit antennas with indices $1$ to $\frac{M_t}{K}$, and so on, which yields $\tilde{\mv{H}}_k=\left[\tilde{\mv{h}}_{\frac{(k-1)M_t}{K}+1},...,\tilde{\mv{h}}_{\frac{kM_t}{K}}\right]$, with $\tilde{\mv{h}}_i$ denoting the $i$th column vector of $\tilde{\mv{H}}$.\footnote{Note that the results are directly extendible to other transmit antenna grouping cases, which will be considered later in Section IV.} Let ${\mv{x}}_k=\left[x_{\frac{(k-1)M_t}{K}+1},...,x_{\frac{k M_t}{K}}\right]^T$ denote the transmitted signal vector for the $k$th group. (\ref{SMchannel}) can be thus rewritten as
\vspace{-1mm}\begin{align}
y_k={\mv{u}}_k^H\tilde{\mv{H}}_k{\mv{x}}_k+{\mv{u}}_k^H\tilde{\mv{H}}_{[-k]}{\mv{x}}_{[-k]}+n_k,\label{SMchannel_new}
\end{align}
where $\tilde{\mv{H}}_{[-k]}=[\tilde{\mv{H}}_1,...,\tilde{\mv{H}}_{k-1},\tilde{\mv{H}}_{k+1},...,\tilde{\mv{H}}_K]$, ${\mv{x}}_{[-k]}=[{\mv{x}}_1^T,...,{\mv{x}}_{k-1}^T,{\mv{x}}_{k+1}^T,...,{\mv{x}}_K^T]^T$. Note that the second term at the right-hand side (RHS) of (\ref{SMchannel_new}) denotes the interference at the $k$th data stream from non-intended transmit antenna groups.

With (\ref{SMchannel_new}), we redefine $d_k={\mv{u}}_k^H\tilde{\mv{H}}_k{\mv{x}}_k$ as the interference-plus-noise-free received signal at the $k$th data stream. The feasible region of $d_k$ is denoted by $\mathcal{D}_k({\mv{u}}_k)$, which is similarly defined as (\ref{BFregion}) for the case of $K=1$. Notice that the set of $d_k$'s as well as the set of $\mathcal{D}_k({\mv{u}}_k)$'s are now decoupled. Therefore, by following similar procedures as in the previous single-stream case, the feasibilities of $\mathcal{S}$ for the $K$ data streams can be separately verified based on $\mathcal{D}_k({\mv{u}}_k)$'s. In addition, given any feasible $\mathcal{S}$ for the $K$ data streams and the corresponding $\{\alpha_k,s_k\}_{k=1}^K$, Problem (\ref{SM_S}) can now be solved by finding each ${\mv{x}}_k$ that yields $d_k=\alpha_ks_k$ separately for all $k$'s; thus (\ref{SMchannel_new}) is equivalently represented by
\vspace{-1mm}\begin{align}
y_k=\alpha_ks_k+{\mv{u}}_k^H\tilde{\mv{H}}_{[-k]}{\mv{x}}_{[-k]}+n_k,\quad s_k\in \mathcal{S}.\label{SMchannel_new2}
\end{align}
Moreover, there are in general two design criteria for $\{{\mv{u}}_k\}_{k=1}^K$ depending on how the interference term in (\ref{SMchannel_new2}) is treated, namely, \emph{MMSE} and \emph{ZF}. Details of the MMSE and ZF based receive beamforming will be presented in Section IV, where the joint design of the transmit antenna grouping and receive beamforming will be addressed as well.
\vspace{-2mm}
\section{Receiver Optimization for Single-Stream Transmission}
\subsection{Problem Formulation}
For single-stream transmission, our objective is to minimize the SER at the combiner output by optimizing the receive beamforming vector ${\mv{u}}$ for given $\tilde{\mv{H}}$ and $\mathcal{S}$. Note that since minimizing the exact SER, $P_s$, is in general a difficult problem, we aim to minimize its union bound instead. We assume $\mathcal{S}$ is an equiprobable signal set and maximum likelihood (ML) detection is used at the combiner output to recover the signal point in $\mathcal{S}$. Without loss of generality, we further assume $\underset{s\in \mathcal{S}}{\max}|s|=1$ for the rest of this paper. The union bound of $P_s$ is thus given by
\vspace{-1mm}\begin{equation}\label{unionbound}
P_s\leq (N-1)Q\left(\sqrt{\frac{(d_\mathrm{min}^\mathrm{c})^2}{2\sigma^2}}\right),
\end{equation}
where $d_{\mathrm{min}}^\mathrm{c}=R({\mv{u}})d_{\mathrm{min}}$ denotes the MED between any two signal points in the scaled constellation $R({\mv{u}})\mathcal{S}$ at the combiner output, with $d_{\mathrm{min}}$ denoting the MED of $\mathcal{S}$ \cite{digicom}. As can be observed from (\ref{unionbound}), minimizing the union bound of $P_s$ is equivalent to maximizing $d_{\mathrm{min}}^\mathrm{c}$, for which we formulate the following optimization problem with given $\tilde{\mv{H}}$ and $\mathcal{S}$ as
\vspace{-3mm}\begin{align}
(\mbox{P1})\quad \underset{{\mv{u}}}{\max}\quad &\|{\mv{u}}^H\tilde{\mv{H}}\|_1\\
\mathrm{s.t.}\quad & \|{\mv{u}}\|_2=1\\ &\frac{\max\left\{2\|{\mv{u}}^H\tilde{\mv{H}}\|_\infty-\|{\mv{u}}^H\tilde{\mv{H}}\|_1,0\right\}}{\|{\mv{u}}^H\tilde{\mv{H}}\|_1}\leq \tau,\label{P1feas}
\end{align}
where $\tau=\underset{s\in \mathcal{S}}{\min} |s|\in [0,1]$, and the feasibility constraint of $\mathcal{S}$ given in (\ref{feascond}) is explicitly expressed in (\ref{P1feas}).

Problem (P1) can be equivalently rewritten as
\vspace{-3mm}\begin{align}
(\mbox{P2})\quad \underset{{\mv{u}}}{\max}\quad &\|{\mv{u}}^H\tilde{\mv{H}}\|_1\\
\mathrm{s.t.}\quad & \|{\mv{u}}\|_2\leq1\label{P2c1}\\
& \|{\mv{u}}^H\tilde{\mv{H}}\|_\infty\leq \frac{\tau+1}{2}\|{\mv{u}}^H\tilde{\mv{H}}\|_1\label{P2c2}\\
& \|{\mv{u}}^H\tilde{\mv{H}}\|_1>0,\label{P2c3}
\end{align}
since it can be shown that ${\mv{u}}^\star$ is optimal for Problem (P1) if and only if ${\mv{u}}^\star$ is the optimal solution to Problem (P2), by noting that the constraint in (\ref{P2c1}) must be satisfied with equality by the optimal solution to Problem (P2).

Note that Problem (P2) is a non-convex optimization problem since the constraints in (\ref{P2c2}) and (\ref{P2c3}) are non-convex. It is also worth noting that Problem (P2) without the constraint in (\ref{P2c2}) can be shown to be equivalent to the class of unimodular quadratic programs (UQPs) that are known to be NP-hard \cite{Huang06}. Moreover, it is non-trivial to extend the existing approaches for finding approximate solutions to the UQPs (e.g., algorithms based on SDR \cite{Huang06,Luo10} or fixed-point iterations \cite{Stoica14}) to the case of Problem (P2), due to the new non-convex constraint in (\ref{P2c2}). As a result, Problem (P2) is in general a difficult problem to solve.

In the following, we first study the feasibility of Problem (P2). Then, we provide an efficient algorithm based on SDR to find an approximate solution for this problem.
\vspace{-3mm}
\subsection{Feasibility of Problem (P2)}
The feasibility of Problem (P2) can be verified by solving the following problem:
\vspace{-2mm}\begin{align}
(\mbox{P2-F})\quad \mathrm{find}&\quad{\mv{u}}\\
\mathrm{s.t.}&\quad \|{\mv{u}}^H\tilde{\mv{H}}\|_\infty\leq \frac{\tau+1}{2}\|{\mv{u}}^H\tilde{\mv{H}}\|_1\label{P2f1}\\
&\quad  \|{\mv{u}}^H{\tilde{\mv{H}}}\|_1>0. \label{P2f2}
\end{align}
Specifically, any feasible solution to Problem (P2) is also a feasible solution to Problem (P2-F); on the other hand, for any feasible solution ${\mv{u}}$ to Problem (P2-F), $\frac{{\mv{u}}}{\|{\mv{u}}\|_2}$ is a feasible solution to Problem (P2).
Although Problem (P2-F) is in general difficult to solve due to the non-convex constraints, useful insights can be drawn by investigating its structure, as shown in the following proposition.

\begin{proposition}\label{prop_BFfeas}
Problem (P2) is feasible if $\mathrm{rank}(\tilde{\mv{H}})\geq 2$.
\end{proposition}
\begin{proof}
Please refer to Appendix \ref{proof_prop_BFfeas}.
\end{proof}

Based on Proposition \ref{prop_BFfeas}, Problem (P2) is always feasible under the assumed i.i.d. Rayleigh fading MIMO channel with $M_t,M_r\geq 2$. Moreover, we provide the following lemma.

\begin{lemma}\label{lemma_P2feas}
The optimal ${\mv{u}}_f$ to the following problem is a feasible solution to Problem (P2):
\begin{align}
(\mbox{P2-FS})\quad \underset{\scriptstyle\|{\mv{u}}_f\|_2\leq 1\atop \scriptstyle j\in \{1,...,M_t\}}{\max}\quad &\mathfrak{Re}\left\{{\mv{u}}_f^H\sum_{i=1}^{M_t}\tilde{\mv{h}}_i\right\}\\
\mathrm{s.t.}\quad &{\mv{u}}_f^H\left(2\tilde{\mv{h}}_j-\sum_{i=1}^{M_t}\tilde{\mv{h}}_i\right)=0.\label{P2FS1}
\end{align}\vspace{-1mm}
\end{lemma}
\begin{proof}
Please refer to Appendix \ref{proof_lemma_P2feas}.
\end{proof}

Based on Lemma \ref{lemma_P2feas}, a feasible solution to Problem (P2) can be obtained by solving Problem (P2-FS). Note that Problem (P2-FS) with any given $j\in\{1,...,M_t\}$ is a convex optimization problem, which can be efficiently solved via existing software, e.g., CVX \cite{cvx}. Thus, by first finding the optimal ${\mv{u}}_f$ for each given $j$, the globally optimal ${\mv{u}}_f$ can be easily obtained via one-dimensional search over $j$.

\subsection{Proposed Solution to Problem (P2)}
In the following, we study how to solve Problem (P2). First, we introduce an auxiliary vector ${\mv{p}}\in \mathbb{C}^{M_t\times1}$ with $|p_i|=1,\ i=1,...,M_t$. The objective function of Problem (P2) can be shown to be equivalently given by
\begin{equation}
\|{\mv{u}}^H\tilde{\mv{H}}\|_1=\underset{\ |p_i|=1,\ i=1,...,M_t}{\max}\ \mathfrak{Re}\left\{{\mv{u}}^H\tilde{\mv{H}}\mv{p}\right\}.\label{equiobj}
\end{equation}
For any given ${\mv{u}}$, denote ${\mv{p}}^\star({\mv{u}})$ as the optimal solution to the problem on the RHS of (\ref{equiobj}), whose elements can be shown to be given by
\begin{align}
p_i^\star({\mv{u}})=e^{-j\arg\left\{{\mv{u}}^H\tilde{\mv{h}}_i\right\}},\quad i=1,...,M_t.\label{optp}
\end{align}

With (\ref{equiobj}) and (\ref{optp}), we have the following proposition.
\begin{proposition}\label{prop_P2equi}
Problem (P2) is equivalent to the following problem:
\begin{align}
\mbox{(P3)}\quad\underset{{\mv{u}},{\mv{p}}}{\max}\quad &\mathfrak{Re}\left\{{\mv{u}}^H\tilde{\mv{H}}\mv{p}\right\}\label{P3obj}\\
\mathrm{s.t.}\quad &\|{\mv{u}}\|_2\leq 1\label{P3c1}\\
&\|{\mv{u}}^H\tilde{\mv{H}}\|_\infty\leq\frac{\tau+1}{2}\mathfrak{Re}\left\{{\mv{u}}^H\tilde{\mv{H}}{\mv{p}}\right\}\label{P3c2}\\
&|p_i|=1,\quad i=1,...,M_t.\label{P3c3}
\end{align}
\end{proposition}

\begin{proof}
Please refer to Appendix \ref{proof_prop_P2equi}.
\end{proof}

Problem (P3) can be shown to be a non-convex QCQP. Next, we propose a customized SDR-based algorithm for solving it. Specifically, we define ${\mv{w}}=[{\mv{u}}^T\ {\mv{p}}^T]^T\in \mathbb{C}^{(M_r+M_t)\times1}$, ${\mv{W}}={\mv{ww}}^H$ and formulate the following problem:
\begin{align}
\mbox{(P3-SDR)}
\quad\underset{{\mv{W}}}{\max}\quad &\mathfrak{Re}\{\mathrm{tr}(\mv{AWG})\}\\
\mathrm{s.t.}\quad& \mathrm{tr}({\mv{W}})\leq M_t+1\\
&\left|{\mv{e}}_i^T{\mv{AWg}}_i\right|\leq \frac{\tau+1}{2}\mathfrak{Re}\{\mathrm{tr}(\mv{AWG})\},\quad i=1,...,M_t\label{P3SDRc1}\\
&[\mv{W}]_{i,i}=1,\quad\ \ i=M_r+1,...,M_r+M_t\\
& \mv{W}\succeq \mv{0},
\end{align}
where ${\mv{A}}=[\mv{0}_{M_t\times M_r}\ \mv{I}_{M_t}]$; $\mv{G}=\left[\tilde{\mv{H}}^T\ {\mv{0}}^T_{M_t\times M_t}\right]^T$ with the $i$th column vector denoted by ${\mv{g}}_i$; ${\mv{e}}_i$ denotes the $i$th column vector of $\mv{I}_{M_t}$. It can be shown that Problem (P3) is equivalent to Problem (P3-SDR) with the additional constraint of $\mathrm{rank}(\mv{W})=1$. Therefore, the optimal value of Problem (P3-SDR) is in general an upper bound on those of Problem (P3) and Problem (P2).

Problem (P3-SDR) is a semi-definite program (SDP), which can be efficiently solved via existing software, e.g. CVX \cite{cvx}. Let ${\mv{W}}^\star$ and $({\mv{u}}^\star,{\mv{p}}^\star)$ denote the optimal solutions to Problem (P3-SDR) and Problem (P3), respectively. If $\mathrm{rank}(\mv{W}^\star)=1$, our relaxation is tight and ${\mv{w}}^\star=[{\mv{u}}^{\star T}\ {\mv{p}}^{\star T}]^T$ can be obtained from the eigenvalue decomposition (EVD) of $\mv{W}^\star$. The optimal solution to Problem (P2) is thus obtained as ${\mv{u}}^\star$. Otherwise, for the case of $\mathrm{rank}(\mv{W}^\star)>1$, we aim to extract an approximate solution to Problem (P2) from ${\mv{W}}^\star$, for which a commonly adopted approach is via the so-called Gaussian randomization method (see e.g., \cite{Luo10} and references therein). By customizing this method to our problem, we propose two randomization algorithms denoted by $\mathrm{Rand}_u$ and $\mathrm{Rand}_p$, which are summarized in Algorithm \ref{randu} and Algorithm \ref{randp}, respectively.
\begin{algorithm}[!htb]
\caption{$\mathrm{Rand}_u$}\label{randu}
\SetKwData{Index}{Index}
\KwIn{$\mv{W}^\star$, $\tilde{\mv{H}}$, $\tau$, $L_u$}
\KwOut{$\tilde{\mv{{u}}}_u$}
Obtain ${\mv{W}}_u^\star\in \mathbb{C}^{M_r\times M_r}$ by $[{\mv{W}}_u^\star]_{i,j}=[{\mv{W}^\star}]_{i,j},\ i=1,...,M_r,\ j=1,...,M_r$.

\eIf{$\mathrm{rank}({\mv{W}}_u^\star)=1$}{Obtain $\tilde{\mv{u}}_u$ by ${\mv{W}}_u^\star=\tilde{\mv{u}}_u\tilde{\mv{u}}_u^H$.\\
\If{$\tilde{\mv{u}}_u$ does not satisfy (\ref{P2c2}) or (\ref{P2c3})}{$\tilde{\mv{u}}_u={\mv{0}}$.}}
{\For{$l=1$ \KwTo $L_u$}{
Generate $\tilde{\mv{v}}^{(l)}\sim \mathcal{CN}(\mv{0},{\mv{W}}_u^\star)$.\\
Obtain $\tilde{\mv{u}}^{(l)}=\frac{\tilde{\mv{v}}^{(l)}}{\|\tilde{\mv{v}}^{(l)}\|_2}$.\\
\If{$\tilde{\mv{u}}^{(l)}$ does not satisfy (\ref{P2c2}) or (\ref{P2c3})}{$\tilde{\mv{u}}^{(l)}={\mv{0}}$.}
}

Set $l^\star=\underset{l=1,...,L_u}{\arg\max}\ \|\tilde{\mv{u}}^{(l)H}\tilde{\mv{H}}\|_1$,
$\tilde{\mv{u}}_u=\tilde{\mv{u}}^{(l^\star)}$.}
\end{algorithm}

\newpage
\begin{algorithm}[t]
\caption{$\mathrm{Rand}_p$}\label{randp}
\SetKwData{Index}{Index}
\KwIn{${\mv{W}}^\star$, $\tilde{\mv{H}}$, $\tau$, $L_p$}
\KwOut{$\tilde{\mv{u}}_p$}
Obtain ${\mv{W}}_p^\star\in \mathbb{C}^{M_t\times M_t}$ by $[{\mv{W}}_p^\star]_{i,j}=[{\mv{W}^\star}]_{M_r+i,M_r+j},\ i=1,...,M_t,\ j=1,...,M_t$.

\eIf{$\mathrm{rank}({\mv{W}}_p^\star)=1$}{Obtain $\tilde{\mv{{p}}}$ by ${\mv{W}}_p^\star=\tilde{\mv{{p}}}\tilde{\mv{{p}}}^H$.\\
\eIf{Problem (P3) is infeasible with given ${\mv{p}}=\tilde{\mv{{p}}}$}{$\tilde{\mv{u}}_p={\mv{0}}$.}{
Obtain $\tilde{\mv{u}}_p$ as the optimal solution to Problem (P3) with given ${\mv{p}}=\tilde{\mv{{p}}}$.}}{
\For{$l=1$ \KwTo $L_p$}{
Generate ${\mv{\xi}}^{(l)}\sim \mathcal{CN}(\mv{0},\mv{W}_p^\star)$.\\
Obtain $\tilde{\mv{{p}}}^{(l)}=\left[\tilde{p}_1^{(l)},...,\tilde{p}_{M_t}^{(l)}\right]^T$ by $\tilde{p}_i^{(l)}=e^{j\arg \left\{\xi^{(l)}_i\right\}},\ \forall i$.

\eIf{Problem (P3) is infeasible with given ${\mv{p}}=\tilde{\mv{{p}}}^{(l)}$}{$\tilde{\mv{u}}^{(l)}={\mv{0}}$.}{
Obtain $\tilde{\mv{u}}^{(l)}$ as the optimal solution to Problem (P3) with given ${\mv{p}}=\tilde{\mv{{p}}}^{(l)}$.}
}

Set $l^\star=\underset{l=1,...,L_p}{\arg\max}\ \|\tilde{\mv{{u}}}^{(l)H}\tilde{\mv{H}}\|_1$,
$\tilde{\mv{u}}_p=\tilde{\mv{u}}^{(l^\star)}$.}
\end{algorithm}

However, it is worth noting that due to the non-convex constraint given in (\ref{P2c2}) of Problem (P2), the feasibility of the approximate solution obtained by Algorithm \ref{randu} or Algorithm \ref{randp} cannot be guaranteed in general (i.e., $\tilde{\mv{u}}_u={\mv{0}}$ or $\tilde{\mv{u}}_p={\mv{0}}$ may occur). Therefore, we propose to employ both Algorithm \ref{randu} and Algorithm \ref{randp} to find $\tilde{\mv{u}}_u$ and $\tilde{\mv{u}}_p$, respectively, based on ${\mv{W}}^\star$; while we also solve Problem (P2-FS) to find a feasible solution denoted by $\tilde{\mv{u}}_f$. Then, an approximate solution to Problem (P2) is chosen from $\tilde{\mv{u}}_u$, $\tilde{\mv{u}}_p$ and $\tilde{\mv{u}}_f$ as the one that achieves the maximum objective value of Problem (P2). It is worth noting that since $\tilde{\mv{u}}_f$ is always a feasible solution to Problem (P2), the feasibility of the selected solution is guaranteed.

To summarize, we provide Algorithm \ref{algoP2}, which finds an approximate solution to Problem (P2) denoted by $\tilde{\mv{u}}$. Note that $\tilde{\mv{{u}}}$ is always feasible for Problem (P2), and is optimal if its corresponding $\mathrm{rank}(\mv{W}^\star)=1$.

\begin{algorithm}[!htb]
\caption{Algorithm for finding an approximate solution to Problem (P2)}\label{algoP2}
\SetKwData{Index}{Index}
\KwIn{$\tilde{\mv{H}}$, $\tau$, $L_u$, $L_p$}
\KwOut{$\tilde{\mv{{u}}}$}

Obtain ${\mv{W}}^\star$ by solving Problem (P3-SDR).

\eIf{$\mathrm{rank}(\mv{W}^\star)=1$}
{Obtain ${\mv{w}}^\star$ by $\mv{W}^\star={\mv{w}}^\star{\mv{w}}^{\star H}$.\\
Obtain $\tilde{\mv{{u}}}={\mv{{u}}}^\star$ by $u_j^\star=w_j^\star,\ j=1,...,M_r$.}
{Obtain $\tilde{\mv{u}}_f$ as the optimal solution to Problem (P2-FS). Obtain $\tilde{\mv{u}}_u$ and $\tilde{\mv{u}}_p$ via Algorithm \ref{randu} and Algorithm \ref{randp}, respectively.\\
Obtain $\tilde{\mv{u}}=\underset{\tilde{\mv{u}}_f,\tilde{\mv{u}}_u,\tilde{\mv{u}}_p}{\arg\max}\left\{\|\tilde{\mv{u}}_f^H\tilde{\mv{H}}\|_1,\|\tilde{\mv{u}}_u^H\tilde{\mv{H}}\|_1,\|\tilde{\mv{u}}_p^H\tilde{\mv{H}}\|_1\right\}$.}
\end{algorithm}
\section{Transceiver Optimization for Multi-Stream Transmission}
For multi-stream transmission, our objective is to minimize the maximum SER over the $K$ data streams, by jointly optimizing $\{{\mv{u}}_k\}_{k=1}^K$ and the transmit antenna grouping with given $\tilde{\mv{H}}$ and $\mathcal{S}$. First, we consider the optimization of $\{{\mv{u}}_k\}_{k=1}^K$ with given transmit antenna grouping, based on the MMSE or ZF criterion, respectively.
\subsection{MMSE-based Receive Beamforming}
\subsubsection{Problem Formulation}
First, note that the inter-group interference at each of the $k$th data stream given by ${\mv{u}}_k^H\tilde{\mv{H}}_{[-k]}{\mv{x}}_{[-k]}={\mv{u}}_k^H\sum_{j\neq k}\tilde{\mv{H}}_{j}{\mv{x}}_{j}$ is a random variable, whose distribution is difficult to obtain since each ${\mv{x}}_j$ is designed via a nonlinear mapping from $\alpha_js_j$ based on ${\mv{u}}_j^H\tilde{\mv{H}}_j$, as illustrated in Section II. For the purpose of analysis, in this subsection, we approximate the inter-group interference at the $k$th data stream by a Gaussian random variable with zero mean and variance $\mathbb{E}[\|{\mv{u}}_k^H\tilde{\mv{H}}_{[-k]}{\mv{x}}_{[-k]}\|_2^2]$.\footnote{Note that as the total number of interfering data streams $K-1$ grows, the accuracy of this approximation increases due to the central limit theorem.}

Then, similar to the case of single-stream transmission, we approximate the SER at the $k$th data stream, $P_s(k)$, with an upper bound, which is given by
\begin{align}
P_s(k)&\overset{(a)}{\leq}(N-1)Q\left(\frac{\sqrt{\frac{P}{M_t}}\|{\mv{u}}_k^H\tilde{\mv{H}}_k\|_1d_\mathrm{min}}{\sqrt{2\left(\mathbb{E}[\|{\mv{u}}_k^H\tilde{\mv{H}}_{[-k]}{\mv{x}}_{[-k]}\|_2^2]+\sigma^2\right)}}\right)\nonumber\\
&\overset{(b)}{\leq}(N-1)Q\left(\frac{\sqrt{\frac{P}{M_t}}\|{\mv{u}}_k^H\tilde{\mv{H}}_k\|_1d_\mathrm{min}}{\sqrt{2\left(\frac{P(K-1)}{K}\|{\mv{u}}_k^H\tilde{\mv{H}}_{[-k]}\|_2^2+\sigma^2\right)}}\right),\label{SERbound}
\end{align}
where $(a)$ results from the SER union bound (recall that we assume $\underset{s\in \mathcal{S}}{\max} |s|=1$, thus we should set $\alpha_k=\sqrt{\frac{P}{M_t}}\|{\mv{u}}_k^H\tilde{\mv{H}}_k\|_1$ in (\ref{SMchannel_new2})); $(b)$ can be derived by noting that $\mathbb{E}[\|{\mv{u}}_k^H\tilde{\mv{H}}_{[-k]}{\mv{x}}_{[-k]}\|_2^2]\leq \mathbb{E}[\|{\mv{u}}_k^H\tilde{\mv{H}}_{[-k]}\|_2^2\|{\mv{x}}_{[-k]}\|_2^2]=\frac{P(K-1)}{K}\|{\mv{u}}_k^H\tilde{\mv{H}}_{[-k]}\|_2^2$ holds due to the Cauchy-Schwarz inequality. Notice from (\ref{SERbound}) that minimizing the maximum SER over the $K$ data streams is equivalent to independently minimizing the SER of each data stream, by solving the following optimization problem for every $k\in \{1,...,K\}$:
\begin{align}
(\mbox{P4})\quad \underset{{\mv{u}}_k}{\max}\quad &\frac{\|{\mv{u}}_k^H\tilde{\mv{H}}_k\|_1}{\sqrt{\|{\mv{u}}_k^H\tilde{\mv{H}}_{[-k]}\|_2^2+\tilde{\sigma}^2}}\\
\mathrm{s.t.}\quad & \|{\mv{u}}_k\|_2\leq1\label{P4c1}\\
& \|{\mv{u}}_k^H\tilde{\mv{H}}_k\|_\infty\leq \frac{\tau+1}{2}\|{\mv{u}}_k^H\tilde{\mv{H}}_k\|_1\label{P4c2}\\
& \|{\mv{u}}_k^H\tilde{\mv{H}}_k\|_1>0,\label{P4c3}
\end{align}
where $\tilde{\sigma}=\sigma\sqrt{\frac{K}{P(K-1)}}$. Note that minimizing the SER upper bound in (\ref{SERbound}) can be shown to be equivalent to minimizing an upper bound of the mean squared error (MSE) between the symbol estimate $\hat{s}_k={\frac{y_k}{\alpha_k}}$ and $s_k$, thus we term this scheme as MMSE-based receive beamforming.
\subsubsection{Feasibility and Proposed Solution of Problem (P4)}
Next, note that Problem (P4) is feasible if and only if there exists ${\mv{u}}_k$ such that the constraints in (\ref{P4c1}), (\ref{P4c2}) and (\ref{P4c3}) are satisfied, which is similar to the feasibility condition of Problem (P2) in the case of single-stream transmission. Therefore, similar to the proof of Proposition \ref{prop_BFfeas}, Problem (P4) can be shown to be feasible if $\mathrm{rank}(\tilde{\mv{H}}_k)\geq 2$, i.e., $\frac{M_t}{K}\geq 2$ under the assumed i.i.d. Rayleigh fading channel.\footnote{It then follows that the number of data streams should satisfy $K\leq \min\left\{M_r,\frac{M_t}{2}\right\}$, which is expected since the degrees of freedom at the transmitter are reduced by half due to the stringent per-antenna CE constraint that fixes the amplitudes of the complex baseband signals.}

Then, by applying the Charnes-Cooper transformation \cite{Charnes} to Problem (P4), we have the following proposition.
\begin{proposition}\label{prop_P5eqv}
Problem (P4) is equivalent to the following problem:
\begin{align}
(\mbox{P5})\quad\underset{{\mv{u}}_k,t_k}{\max}\quad &\|{\mv{u}}_k^H\tilde{\mv{H}}_k\|_1\\
\mathrm{s.t.}\quad & \|{\mv{u}}_k^H\tilde{\mv{H}}_{[-k]}\|_2^2+t_k\tilde{\sigma}^2=1\\
& \|{\mv{u}}_k\|_2\leq \sqrt{t_k}\\
      & \|{\mv{u}}_k^H\tilde{\mv{H}}_k\|_\infty\leq \frac{\tau+1}{2}\|{\mv{u}}_k^H\tilde{\mv{H}}_k\|_1\\
      & \|{\mv{u}}_k^H\tilde{\mv{H}}_k\|_1>0\\
      & t_k>0.
\end{align}
\end{proposition}
\begin{proof}
Please refer to Appendix \ref{proof_prop_P5eqv}.
\end{proof}

Furthermore, by introducing an auxiliary vector ${\mv{p}}_k\in \mathbb{C}^{\frac{M_t}{K}\times 1}$ with each element satisfying $|p_{ki}|=1,\ i=1,...,\frac{M_t}{K}$, Problem (P5) can be shown to be equivalent to the following problem:
\begin{align}
(\mbox{P6})\quad\underset{{\mv{u}}_k,{\mv{p}}_k}{\max}\quad &\mathfrak{Re}\left\{{\mv{u}}_k^H\tilde{\mv{H}}_k{\mv{p}}_k\right\}\\
\mathrm{s.t.}\quad & \tilde{\sigma}^2\|{\mv{u}}_k\|_2^2+\|{\mv{u}}_k^H\tilde{\mv{H}}_{[-k]}\|_2^2\leq1\\
      & \|{\mv{u}}_k^H\tilde{\mv{H}}_k\|_\infty\leq \frac{\tau+1}{2}\mathfrak{Re}\left\{{\mv{u}}_k^H\tilde{\mv{H}}_k{\mv{p}}_k\right\}\\
      & |p_{ki}|=1,\quad i=1,...,\frac{M_t}{K}.
\end{align}
The proof is similar to that of Proposition \ref{prop_P2equi}, thus is omitted here for brevity. It is worth noting that given any feasible ${\mv{u}}_k$ to Problem (P6), the optimal value of Problem (P6) can be shown to be the same as that of Problem (P4) with the solution $\tilde{\sigma}{\mv{u}}_k\big/\sqrt{1-\|{\mv{u}}_k^{H}\tilde{\mv{H}}_{[-k]}\|_2^2}$.

Note that Problem (P6) is a non-convex QCQP. In the following, we apply the SDR technique for finding an approximate solution to Problem (P6) as well as Problem (P4). Specifically, we define ${\mv{b}}_k=[{\mv{u}}_k^T\ {\mv{p}}_k^T]^T$, ${\mv{W}}_k={\mv{b}}_k{\mv{b}}_k^H$ and formulate the following problem:
\begin{align}
(\mbox{P6-SDR})\quad\underset{{\mv{W}}_k}{\max}\quad &\mathfrak{Re}\{\mathrm{tr}\left({\mv{AW}}_k{\mv{G}}_k\right)\}\\
\mathrm{s.t.}\quad & {\tilde{\sigma}^2}\mathrm{tr}({\mv{W}}_k)+\mathrm{tr}\left({\mv{W}}_k{\mv{G}}_{[-k]}{\mv{G}}_{[-k]}^H\right)\leq 1+\frac{M_t}{K}{\tilde{\sigma}^2}\\
      & \left|{\mv{f}}_i^T{\mv{AW}}_k{\mv{g}}_{ki}\right|\leq \frac{\tau+1}{2}\mathfrak{Re}\left\{\mathrm{tr}({\mv{AW}}_k{\mv{G}}_k)\right\},\quad i=1,...,\frac{M_t}{K}\\
      & [{\mv{W}}_k]_{i,i}=1,\quad i=M_r+1,...,M_r+\frac{M_t}{K}\\
      & {\mv{W}}_k\succeq {\mv{0}},
\end{align}
where ${\mv{A}}=\left[{\mv{0}}_{\frac{M_t}{K}\times M_r}\ {\mv{I}}_{\frac{M_t}{K}}\right]$; ${\mv{G}}_k=\left[\tilde{\mv{H}}_k^T\ {\mv{0}}_{\frac{M_t}{K}\times \frac{M_t}{K}}^T\right]^T$ with the $i$th column vector denoted by ${\mv{g}}_{ki}$; ${\mv{G}}_{[-k]}=\left[\tilde{\mv{H}}_{[-k]}^T\ {\mv{0}}_{\frac{M_t}{K}\times \frac{M_t(K-1)}{K}}^T\right]^T$; ${\mv{f}}_i$ denotes the $i$th column vector of ${\mv{I}}_{\frac{M_t}{K}}$. It can be shown that Problem (P6) is equivalent to Problem (P6-SDR) with the additional constraint of $\mathrm{rank}({\mv{W}}_k)=1$, thus the optimal value of Problem (P6-SDR) is in general an upper bound on those of Problems (P6) and (P4).

Problem (P6-SDR) is an SDP which can be efficiently solved via existing software, e.g., CVX \cite{cvx}. Let ${\mv{W}}_k^\star$ and $({\mv{u}}_k^\star,{\mv{p}}_k^\star)$ denote the optimal solutions to Problem (P6-SDR) and Problem (P6), respectively. If $\mathrm{rank}({\mv{W}}_k^\star)=1$, our relaxation is tight, and ${\mv{b}}_k^\star=[{\mv{u}}_k^{\star T}\ {\mv{p}}_k^{\star T}]^T$ can be obtained from the EVD of ${\mv{W}}_k^\star$. The optimal solution to Problem (P4) is thus obtained as $\tilde{\sigma}{\mv{u}}_k^\star\big/\sqrt{1-\|{\mv{u}}_k^{\star H}\tilde{\mv{H}}_{[-k]}\|_2^2}$. Otherwise, for the case of $\mathrm{rank}({\mv{W}}_k^\star)>1$, by noticing the similarity between Problem (P6) and Problem (P3), as well as that between Problem (P6-SDR) and Problem (P3-SDR), an approximate solution to Problem (P4) can be extracted from ${\mv{W}}_k^\star$ by applying the Gaussian randomization methods proposed in Section III (i.e., $\mathrm{Rand}_u$ and $\mathrm{Rand}_p$) with minor modification. It is also worth noting that a feasible solution to Problem (P4) can be always obtained by solving a similar problem as Problem (P2-FS). The overall algorithm for finding an approximate solution to Problem (P4) is similar to Algorithm 3 for Problem (P2) in the case of single-stream transmission, which is thus omitted here for brevity.
\subsection{ZF-based Receive Beamforming}
\subsubsection{Problem Formulation}
In this subsection, we consider the ZF-based receive beamforming, where the inter-group interference at each data stream is eliminated by designing $\{{\mv{u}}_k\}_{k=1}^K$ subject to the following constraints:
\begin{align}
{\mv{u}}_k^H\tilde{\mv{H}}_{[-k]}={\mv{0}},\quad\forall k.\label{ZFconstraint}
\end{align}
Note that the equalities in (\ref{ZFconstraint}) have non-trivial solutions (i.e., ${\mv{u}}_k\neq {\mv{0}},\ \forall k$) if and only if $\mathrm{rank}\left(\tilde{\mv{H}}_{[-k]}\right)=\min\left\{M_r,\frac{M_t(K-1)}{K}\right\}< M_r,\ \forall k$ holds. This implies $M_r\geq \frac{M_t(K-1)}{K}+1$ needs to be true, which is thus assumed in this subsection.\footnote{It is also worth noting that our results can be extended to the case of $M_r< \frac{M_t(K-1)}{K}+1$, by switching off an appropriate number of transmit antennas.}

The structures of ${\mv{u}}_k$'s that satisfy (\ref{ZFconstraint}) can be simplified as follows. Let the singular value decomposition (SVD) of $\tilde{\mv{H}}_{[-k]}^H$ be denoted as
\begin{align}
\tilde{\mv{H}}_{[-k]}^H={\mv{U}}_k{\mv{\Lambda}}_k {\mv{V}}_k^H={\mv{U}}_k{\mv{\Lambda}}_k[\bar{\mv{V}}_k\ \tilde{\mv{V}}_k]^H,
\end{align}
where ${\mv{U}}_k\in \mathbb{C}^{\frac{(K-1)M_t}{K}\times\frac{(K-1)M_t}{K}}$ and ${\mv{V}}_k\in \mathbb{C}^{M_r\times M_r}$ are unitary matrices, i.e., ${\mv{U}}_k{\mv{U}}_k^H={\mv{U}}_k^H{\mv{U}}_k={\mv{I}}_{\frac{(K-1)M_t}{K}}, {\mv{V}}_k{\mv{V}}_k^H={\mv{V}}_k^H{\mv{V}}_k={\mv{I}}_{M_r}$, and ${\mv{\Lambda}}_k=\left[{\mv{\Sigma}}_k\ {\mv{0}}\right]\in \mathbb{C}^{\frac{(K-1)M_t}{K}\times M_r}$ with ${\mv{\Sigma}}_k\in \mathbb{C}^{\frac{(K-1)M_t}{K}\times \frac{(K-1)M_t}{K}}$ being a diagonal matrix. Furthermore, $\bar{\mv{V}}_k\in \mathbb{C}^{M_r\times \frac{(K-1)M_t}{K}}$ and $\tilde{\mv{V}}_k\in \mathbb{C}^{M_r\times \left(M_r-\frac{(K-1)M_t}{K}\right)}$ consist of the first $\frac{(K-1)M_t}{K}$ and the last $M_r-\frac{(K-1)M_t}{K}$ right singular vectors of $\tilde{\mv{H}}_{[-k]}^H$, respectively. It can be shown that $\tilde{\mv{V}}_k$ with $\tilde{\mv{V}}_k^H\tilde{\mv{V}}_k={\mv{I}}_{M_r-\frac{(K-1)M_t}{K}}$ forms an orthogonal basis for the null space of $\tilde{\mv{H}}_{[-k]}^H$. Therefore, to guarantee ${\mv{u}}_k^H\tilde{\mv{H}}_{[-k]}={\mv{0}}$, ${\mv{u}}_k$ must be in the following form:
\begin{align}
{\mv{u}}_k=\tilde{\mv{V}}_k{\mv{w}}_k,\label{uk}
\end{align}
where ${\mv{w}}_k\in \mathbb{C}^{\left(M_r-\frac{(K-1)M_t}{K}\right)\times1}$. By defining ${\mv{H}}_k=\tilde{\mv{V}}_k^H{\tilde{\mv{H}}_k}$, we have
\begin{align}
y_k={\mv{u}}_k^H\tilde{\mv{H}}_k{\mv{x}}_k+n_k={\mv{w}}_k^H{\mv{H}}_k{\mv{x}}_k+n_k.\label{ZFchannel}
\end{align}

As a result of (\ref{ZFchannel}), the $K$ data streams are transmitted over $K$ parallel smaller-size MIMO sub-channels ${\mv{H}}_k$'s, each with receive beamforming vector ${\mv{w}}_k$ and transmitted signal vector ${\mv{x}}_k$. Similar to the case of single-stream transmission, we aim to minimize the maximum union bound of SER over the $K$ data streams, which can be shown to be equivalent to solving the following problem for every $k\in \{1,...,K\}$:
\begin{align}
(\mbox{P7})\quad\underset{{\mv{w}}_k}{\max}\quad &\|{\mv{w}}_k^H{\mv{H}}_k\|_1\\
\mathrm{s.t.}\quad & \|{\mv{w}}_k\|_2\leq1\\
      & \|{\mv{w}}_k^H{\mv{H}}_k\|_\infty\leq \frac{\tau+1}{2}\|{\mv{w}}_k^H{\mv{H}}_k\|_1\label{P4.1c2}\\
      & \|{\mv{w}}_k^H{\mv{H}}_k\|_1>0.
\end{align}
\subsubsection{Feasibility and Proposed Solution of Problem (P7)}
Next, by generalizing the result in Proposition \ref{prop_BFfeas}, we provide a sufficient condition under which the feasibility of Problem (P7) is guaranteed for every $k$, as shown in the following proposition.
\begin{proposition}\label{prop_SMfeas}
Problem (P7) is feasible for all $k\in \{1,...,K\}$ with any given transmit antenna grouping, if $M_r\geq \frac{(K-1)M_t}{K}+2$ and $\frac{M_t}{K}\geq 2$.
\end{proposition}
\begin{proof}
Please refer to Appendix \ref{proof_prop_SMfeas}.
\end{proof}
\begin{remark}
It is worth noting that for the case of $2K\leq M_r< \frac{(K-1)M_t}{K}+2$ and $\frac{M_t}{K}\geq 2$, this scheme can still be made feasible by selecting a subset of $M_t', M_t'\in \left[2K,\frac{K(M_r-2)}{K-1}\right]$ transmit antennas for CE precoding with the other antennas not used.
\end{remark}

Notice that Problem (P7) is in the same form as Problem (P2). Therefore, the solution to Problem (P7) can be readily obtained by applying Algorithm 3.

Finally, by solving Problem (P4) or Problem (P7) for all $k\in \{1,...,K\}$ with $\{\tilde{\mv{H}}_k\}_{k=1}^K$ or $\left\{{\mv{H}}_k\right\}_{k=1}^K$ resulting from all possible transmit antenna groupings, the optimal grouping can be obtained as the one that yields the maximum minimum objective value of Problem (P4) or Problem (P7) over all $k$'s, respectively.
\section{Numerical Results}
In this section, we provide numerical results to corroborate our study. We assume $\mathcal{S}$ is an $N$-ary QAM constellation unless specified otherwise. The average signal-to-noise ratio (SNR) is defined as $\mathrm{SNR}=\frac{P\beta}{\sigma^2}$. The numbers of randomization trials for $\mathrm{Rand}_u$ and $\mathrm{Rand}_p$ are set as $L_u=50$ and $L_p=50$, respectively.
\subsection{Single-Stream Transmission}
In this subsection, we consider the case of single-stream transmission (i.e., $K=1$) and compare the performance of our proposed receive beamforming scheme with the following benchmark schemes.
\begin{itemize}
\item {\bf{Antenna Selection (AS)}}: In this scheme, the $j$th element in the receive beamforming vector is given by $u_j=1$ if $j=j^\star$, and $u_j=0$ otherwise, where $j^\star$ denotes the optimal solution to the following problem:
    \begin{align}
    \underset{j=1,...,M_r}{\max}\quad &\|\tilde{\mv{h}}'_j\|_1\label{PBS1}\\
    \mathrm{s.t.} \quad &\|\tilde{\mv{h}}'_j\|_\infty\leq \frac{\tau+1}{2}\|\tilde{\mv{h}}'_j\|_1,\nonumber
    \end{align}
    where $\tilde{\mv{h}}'_j$ denotes the transposed vector of the $j$th row of $\tilde{\mv{H}}$. Problem (\ref{PBS1}) can be easily solved via one-dimensional search over $j$. If Problem (\ref{PBS1}) is infeasible, we set $j^\star=\underset{j=1,...,M_r}{\arg\max}\ \|\tilde{\mv{h}}'_j\|_1$.
\item {\bf{Strongest Eigenmode Beamforming (SEB)}}: In this scheme, the receive beamforming vector is obtained as the optimal solution to the following problem:
    \begin{align}
    \underset{\|{\mv{u}}\|_2\leq1}{\max}\quad &\|{\mv{u}}^H{\tilde{\mv{H}}}\|_2,\label{PBS2}
    \end{align}
    which can be shown to be the eigenvector corresponding to the maximum eigenvalue of ${{\tilde{\mv{H}}\tilde{\mv{H}}}}^H$.
\end{itemize}

In Fig. \ref{BFSER}, we consider the case of $N=16$ and show the average SERs of our proposed scheme and the benchmark schemes under the following setups: i) $M_t=2, M_r=4$ and ii) $M_t=M_r=4$. Note that for AS or SEB with $16$-QAM constellation (i.e., $\underset{s\in \mathcal{S}}{\min}|s|=\tau=\frac{1}{3}$), high SER can occur if the resulting receive beamforming vector $\mv{u}$ does not satisfy the constraint in (\ref{P2c2}), thus is infeasible for CE precoding. Therefore, we also show in Fig. \ref{BFSER} the average SERs of AS and SEB with hybrid $16$-QAM/$16$-PSK (phase shift keying) constellations, where the constellation $\mathcal{S}$ at the combiner output is adaptively switched to $16$-PSK if AS or SEB is infeasible with $16$-QAM, to achieve the same transmission rate. Note that such schemes are always feasible, since $16$-PSK constellation yields $\underset{s\in \mathcal{S}}{\min}|s|=\tau=1$, thus the constraint in (\ref{P2c2}) is always satisfied.
\begin{figure}[t]
  \centering
  \subfigure[$M_t=2,M_r=4$]{
    \label{fig:subfig:a}
    \includegraphics[width=3.12in]{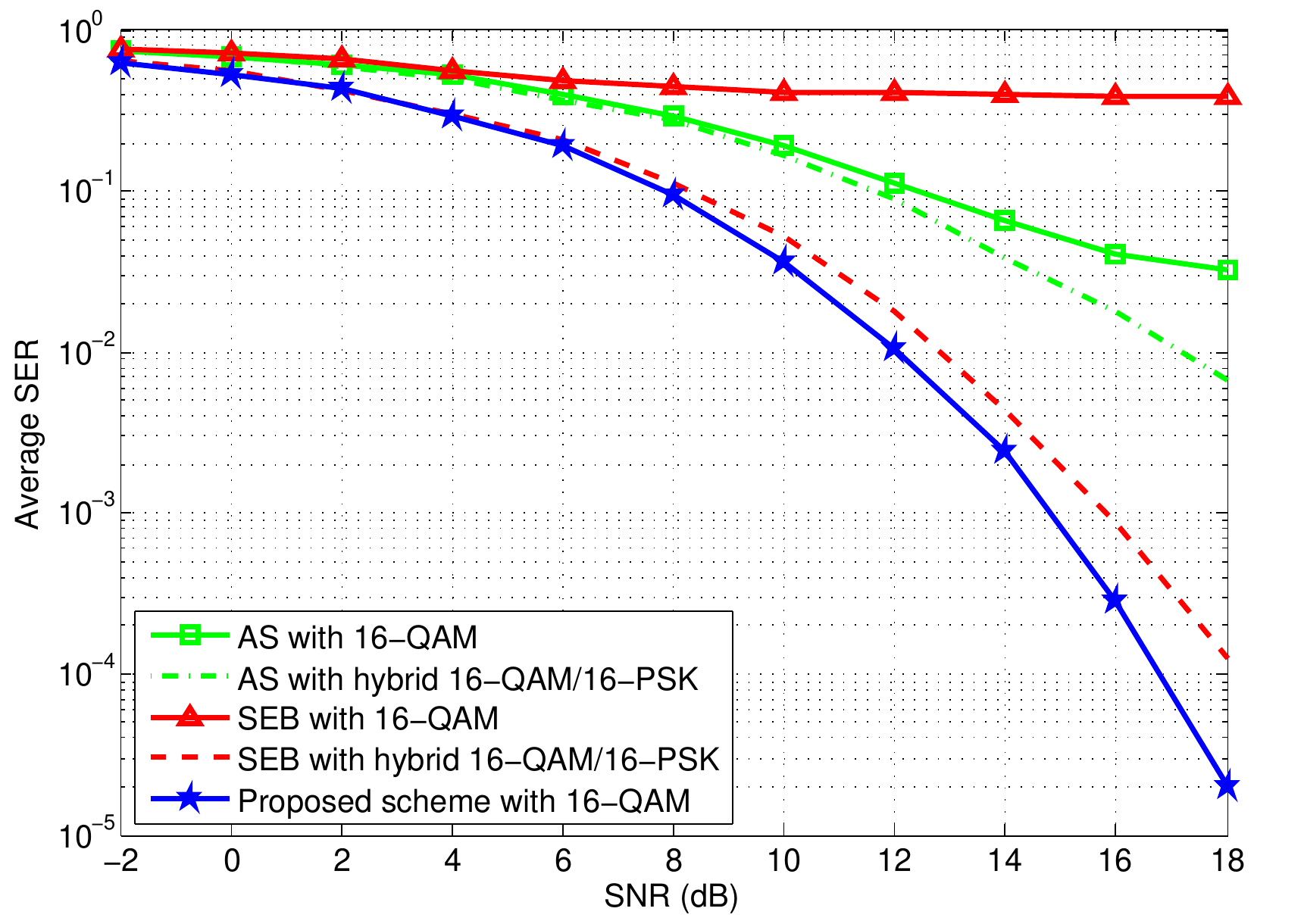}}
  \hspace{0in}
  \subfigure[$M_t=M_r=4$]{
    \label{fig:subfig:b}
    \includegraphics[width=3.12in]{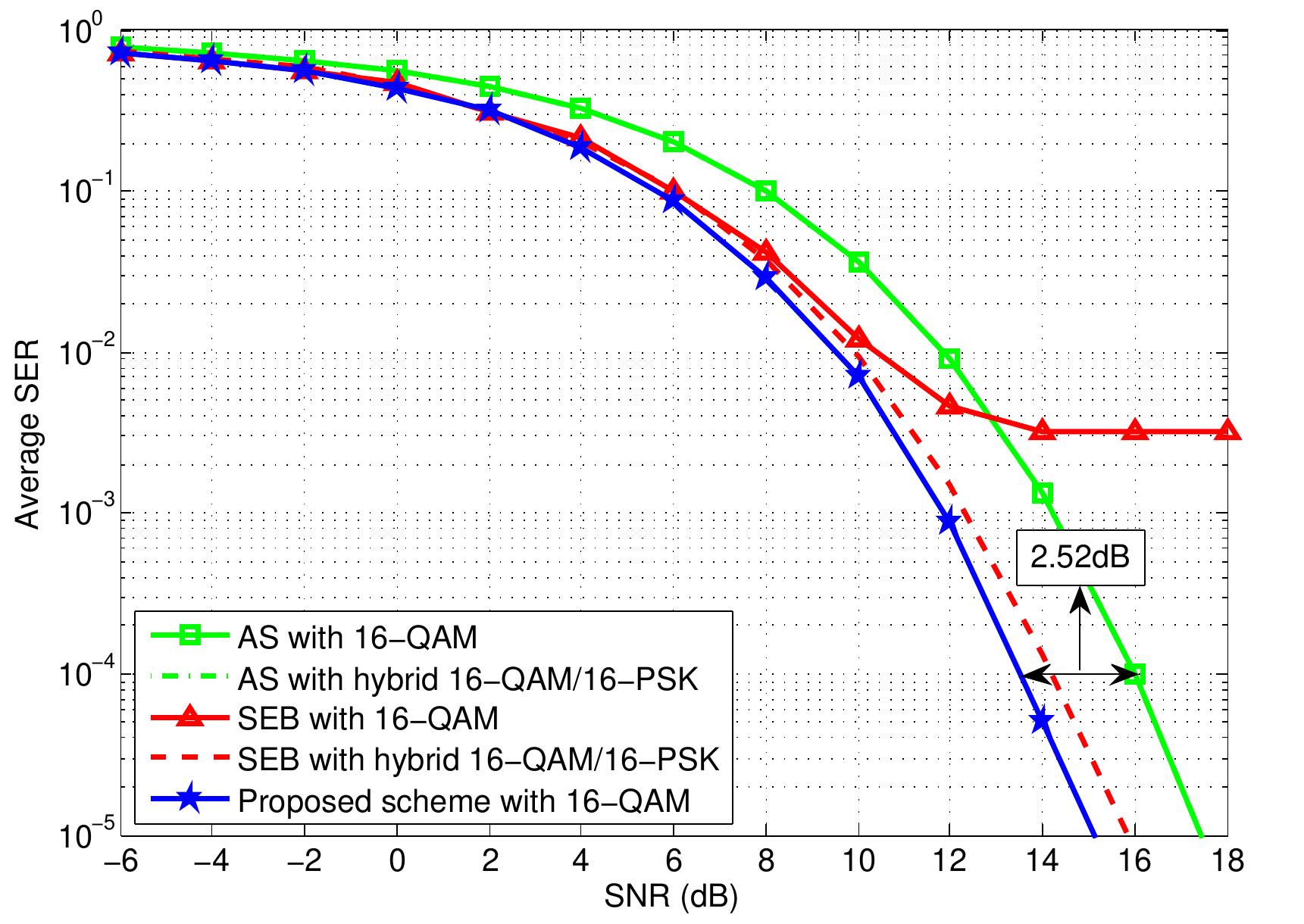}}
  \caption{Average SER comparison of receive beamforming schemes for CE single-stream transmission.}
  \label{BFSER}
\end{figure}

For both setups, it is observed from Fig. \ref{BFSER} that our proposed scheme outperforms both AS and SEB with $16$-QAM. Specifically, AS with $16$-QAM results in error floor for the case of $M_t=2, M_r=4$, and has an SNR loss of $2.52$dB compared with our proposed scheme at the average SER of $10^{-4}$ for the case of $M_t=M_r=4$. On the other hand, SEB with $16$-QAM results in error floor under both setups. Note that the performance gain of our proposed scheme is due to the optimization of $\mv{u}$, as well as the fact that AS and SEB with $16$-QAM may not be always feasible for CE precoding with any channel realization, while our proposed scheme is always feasible (as a consequence of Proposition \ref{prop_BFfeas}). Moreover, it is observed that our proposed scheme outperforms AS and SEB even with hybrid $16$-QAM/$16$-PSK under both setups. This implies that compared to using adaptive receiver constellation which requires additional implementation complexity, our proposed design of receive beamforming is a more cost-effective method for guaranteeing the feasibility of CE single-stream transmission and also achieves better average SER performance.
\subsection{Multi-Stream Transmission}
\begin{figure}[t]
  \centering
  \subfigure[$\bar{R}=2$bps/Hz, $M_t=M_r=4$]{
    \label{fig:subfig:a}
    \includegraphics[width=3.12in]{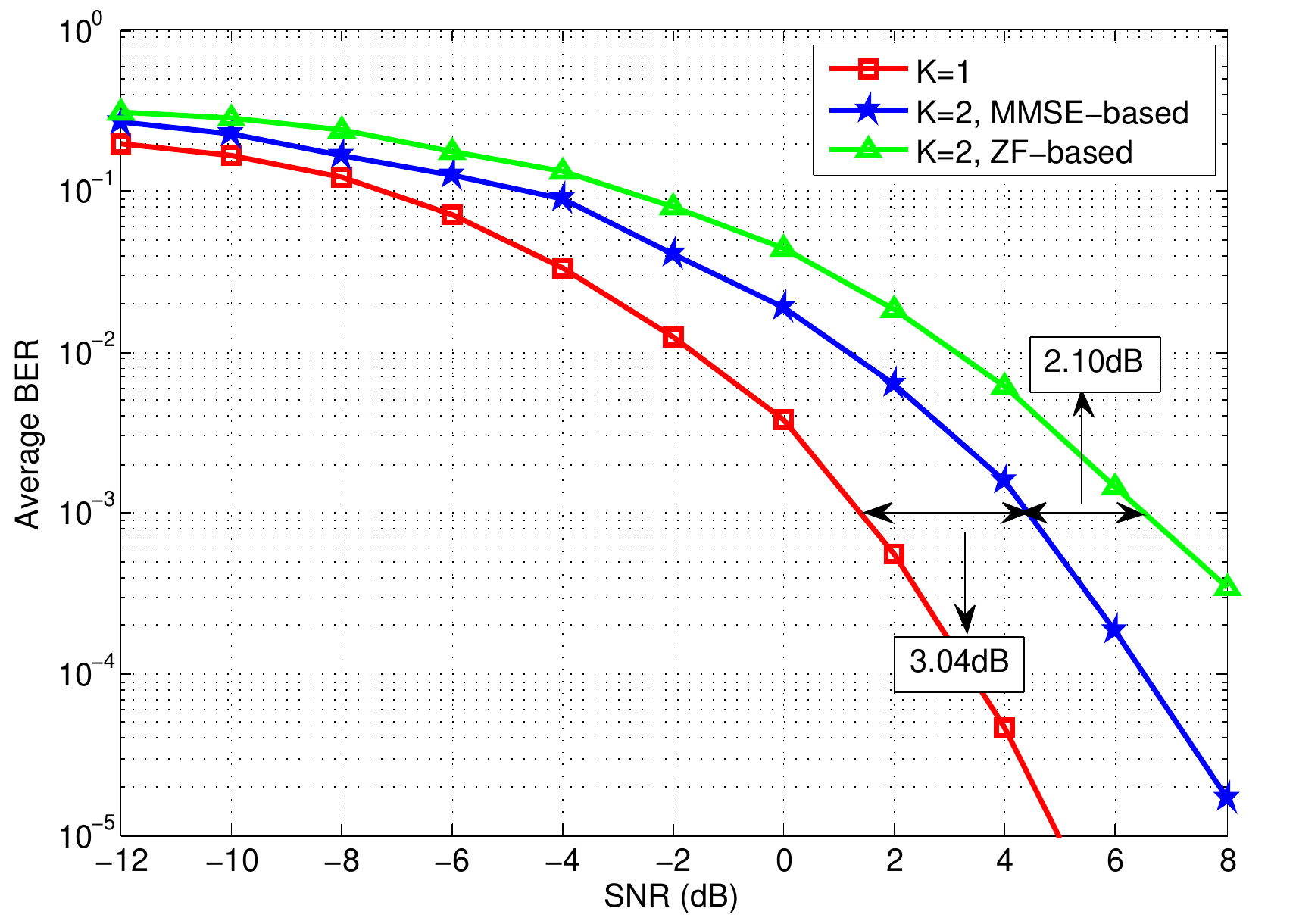}}
    \hspace{0in}
  \subfigure[$\bar{R}=8$bps/Hz, $M_t=M_r=4$]{
    \label{fig:subfig:b}
    \includegraphics[width=3.12in]{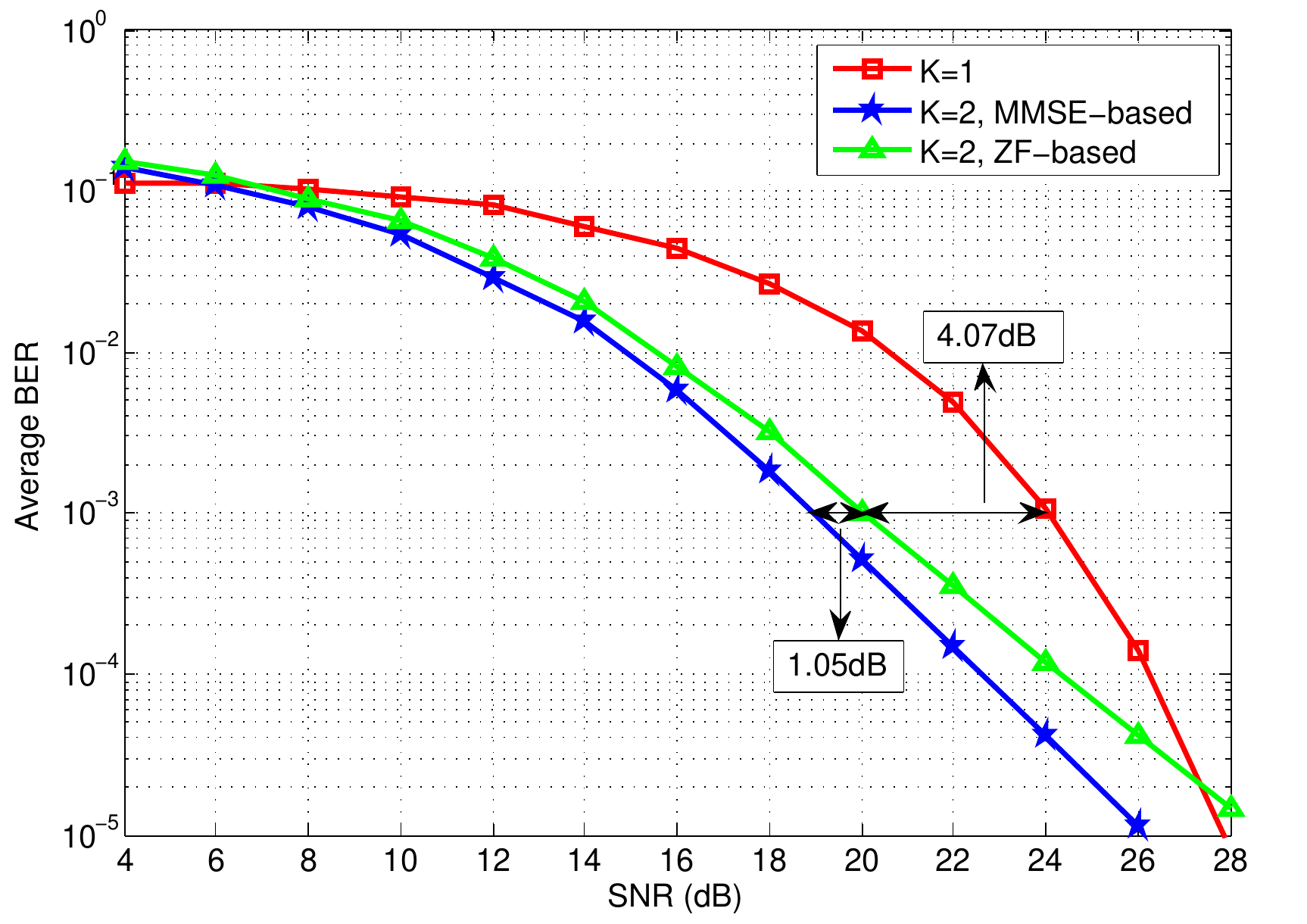}}
    \hspace{0in}
    \subfigure[$\bar{R}=8$bps/Hz, $M_t=M_r=8$]{
    \label{fig:subfig:b}
    \includegraphics[width=3.12in]{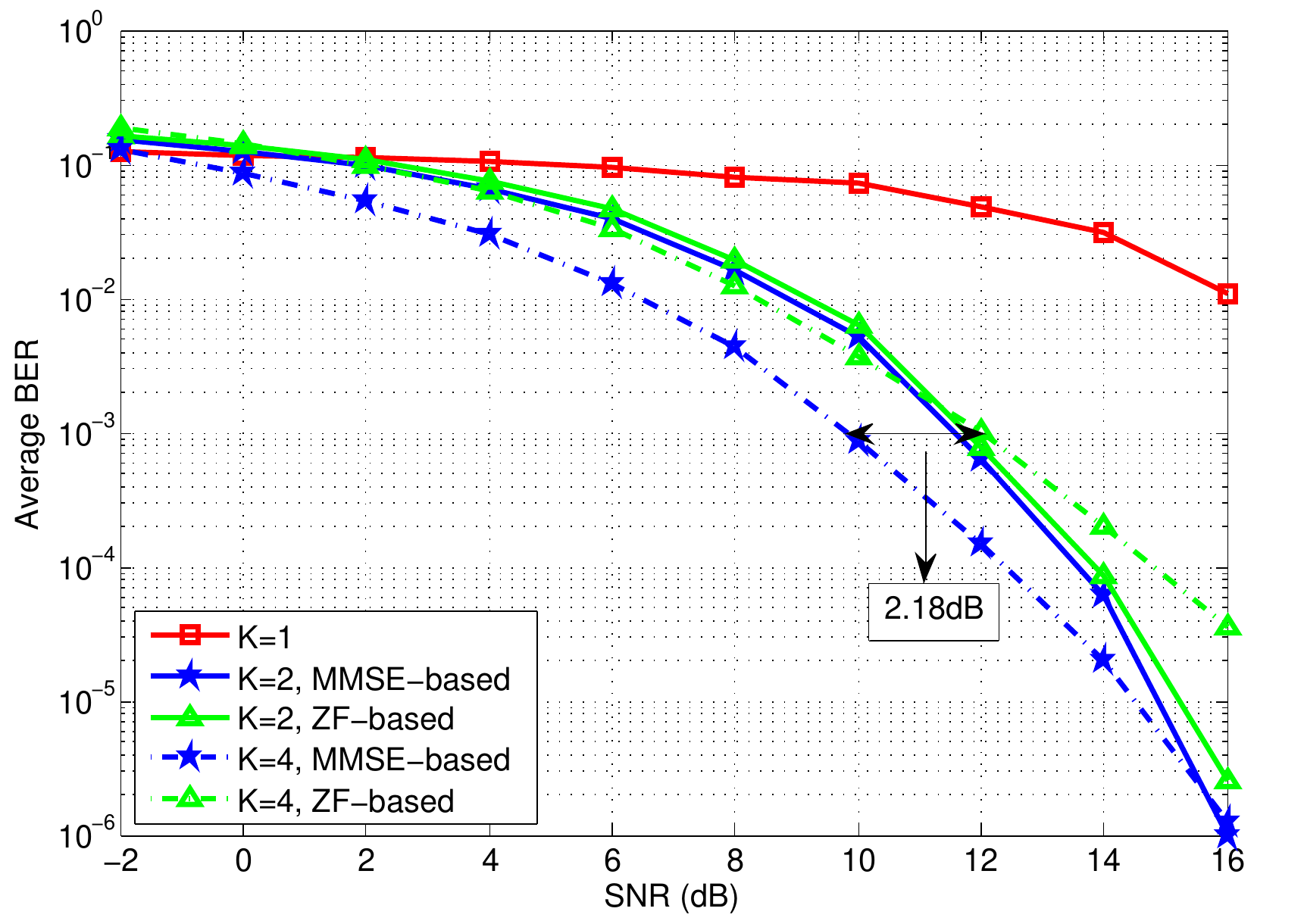}}
  \caption{Average BER comparison of CE multi-stream versus single-stream transmission schemes.}
  \label{SMBF}
\end{figure}
In this subsection, we evaluate the performance of our proposed CE multi-stream transmission schemes. In Fig. \ref{SMBF}, we compare the average bit error rate (BER) of three schemes given the same transmission rate $\bar{R}$: single-stream transmission, MMSE-based and ZF-based multi-stream transmission, respectively, under three setups: i) $\bar{R}=2$bps/Hz, $M_t=M_r=4$; ii) $\bar{R}=8$bps/Hz, $M_t=M_r=4$; and iii) $\bar{R}=8$bps/Hz, $M_t=M_r=8$, respectively. Note that we consider $K=2$ for the MMSE and ZF based schemes under the first two setups with $M_t=M_r=4$, and both $K=2$ and $K=4$ under the third setup with $M_t=M_r=8$, which can be shown to be feasible according to the results in Section IV.

First, it is observed that under all three setups and for any value of $K$, the MMSE-based multi-stream transmission scheme outperforms the ZF-based scheme for all SNR values. Specifically, at the average BER of $10^{-3}$, the SNR gain of the MMSE-based scheme over the ZF-based scheme is $2.10$dB and $1.05$dB for the first two setups, respectively, and $2.18$dB for the third setup with $K=4$. This can be explained by noting that the ZF-based scheme in general yields a suboptimal solution to Problem (P4) for the MMSE-based scheme.

Next, with given transmission rate, we investigate the effect of $K$ on the BER performance in order to draw insights for selecting the optimal transmission mode in practice. From Fig. \ref{SMBF} (a) and (b) with $M_t=M_r=4$, it is observed that at the average BER of $10^{-3}$, the single-stream transmission scheme has an SNR gain of $3.04$dB over the MMSE-based multi-stream scheme for the case of $\bar{R}=2$bps/Hz, but suffers from an SNR loss of $4.07$dB compared to the ZF-based multi-stream scheme for the case of $\bar{R}=8$bps/Hz; moreover, for the case of $\bar{R}=8$bps/Hz, the single-stream scheme eventually outperforms the two multi-stream schemes as the SNR grows. Similarly, it is observed from Fig. \ref{SMBF} (c) that the MMSE and ZF based schemes with $K=4$ outperform those with $K=2$, respectively, in the moderate-SNR regime; while the reverse is true in the high-SNR regime. This reveals that a large $K$ (i.e., multiplexing a large number of data streams) is preferable in the high-rate regime with moderate SNR, by exploiting more multiplexing gain of the MIMO channel; while a small $K$ is suitable for the low-rate and/or high-SNR regime, by extracting more beamforming gain from the MIMO channel.
\section{Conclusion}
This paper investigated the transceiver design for the MIMO channel with CE precoding. For single-stream transmission, we studied the receive beamforming optimization problem for any channel realization and desired constellation at the combiner output, to maximize the MED between any two signal points at the combiner output subject to the feasibility constraint of the constellation. We showed that this problem is always feasible under i.i.d. Rayleigh fading, and proposed an efficient algorithm based on SDR to find an approximate solution. The proposed receive beamforming scheme was shown to significantly outperform other benchmark schemes in terms of average SER. For multi-stream transmission, a new scheme adopting transmit antenna grouping and receive MMSE or ZF based beamforming was proposed. The joint design of the transmit antenna grouping and receive beamforming was further optimized to minimize the maximum SER over all data streams subject to the constellation feasibility constraints. Numerical results showed that the MMSE-based receive beamforming outperforms the ZF-based receive beamforming in terms of average BER; moreover, for fixed transmission rate, it is desirable to transmit with a large number of data streams in the high-rate and moderate-SNR regime, and a small number of data streams in the low-rate and/or high-SNR regime.
\appendices
\section{Proof of Proposition \ref{prop_BFfeas}}\label{proof_prop_BFfeas}
To start, we present the following lemma.
\begin{lemma}\label{lemma_feas}
If $\mathrm{rank}(\tilde{\mv{H}})\geq 2$, there exists a $\bar{\mv{u}}_l\in \mathbb{C}^{M_r\times1}$ for any $l\in \{1,...,M_t\}$ that satisfies the following conditions:
\begin{align}
\bar{\mv{u}}_l^H\left(2\tilde{\mv{h}}_l-\sum_{i=1}^{M_t}\tilde{\mv{h}}_i\right)=&0\label{f1}\\
\|\bar{\mv{u}}_l^H\tilde{\mv{H}}\|_1>&0\label{f3}.
\end{align}
\end{lemma}
\begin{proof}
We prove Lemma \ref{lemma_feas} by contradiction. Suppose any solution $\bar{\mv{u}}_l$ to (\ref{f1}) is also a solution to $\bar{\mv{u}}_l^H\tilde{\mv{H}}={\mv{0}}$, then it can be shown that $\mathrm{Null}\left(2\tilde{\mv{h}}_l^H-\sum_{i=1}^{M_t}\tilde{\mv{h}}_i^H\right)\subseteq \mathrm{Null}\left(\tilde{\mv{H}}^H\right)$. It thus follows that $\mathrm{rank}\left(\tilde{\mv{H}}^H\right)\leq \mathrm{rank}\left(2\tilde{\mv{h}}_l^H-\sum_{i=1}^{M_t}\tilde{\mv{h}}_i^H\right)=1$. This contradicts the assumption of $\mathrm{rank}(\tilde{\mv{H}})\geq 2$. The proof of Lemma \ref{lemma_feas} is thus completed.
\end{proof}

With Lemma \ref{lemma_feas}, we prove Proposition \ref{prop_BFfeas} by showing that for any $l\in\{1,...,M_t\}$, any $\bar{\mv{u}}_l$ that satisfies (\ref{f1}) and (\ref{f3}) is a feasible solution to Problem (P2-F). First, notice that $\bar{\mv{u}}_l$ satisfies the constraint in (\ref{P2f2}). Then, we show that $\bar{\mv{u}}_l$ also satisfies the constraint in (\ref{P2f1}). Specifically, we have
\begin{align}
&|\bar{\mv{u}}_l^H\tilde{\mv{h}}_l|\overset{(a_1)}{=}\left|\bar{\mv{u}}_l^H\sum_{i\neq l}\tilde{\mv{h}}_i\right|\overset{(b_1)}{\leq}\sum_{i\neq l} |\bar{\mv{u}}_l^H\tilde{\mv{h}}_i|\label{p1eq2}\\
&|\bar{\mv{u}}_l^H\tilde{\mv{h}}_j|\overset{(a_2)}{=}\left|\bar{\mv{u}}_l^H\tilde{\mv{h}}_l-\sum_{i\neq l,i\neq j}\bar{\mv{u}}_l^H\tilde{\mv{h}}_i\right|\overset{(b_2)}{\leq} \sum_{i\neq j} |\bar{\mv{u}}_l^H\tilde{\mv{h}}_i|,\quad \forall j\in \{1,...,M_t\}\backslash\{l\},\label{p1eq3}
\end{align}
where ($a_1$) and ($a_2$) result from (\ref{f1}); ($b_1$) and ($b_2$) are due to the triangle inequality. It then follows from (\ref{p1eq2}) and (\ref{p1eq3}) that
\begin{align}
|\bar{\mv{u}}_l^H\tilde{\mv{h}}_i|\leq \|\bar{\mv{u}}_l^H\tilde{\mv{H}}\|_1-|\bar{\mv{u}}_l^H\tilde{\mv{h}}_i|,\quad \forall i\in \{1,...,M_t\},\label{eq3}
\end{align}
namely,
\begin{align}
\|\bar{\mv{u}}_l^H\tilde{\mv{H}}\|_\infty\leq \frac{1}{2}\|\bar{\mv{u}}_l^H\tilde{\mv{H}}\|_1\overset{(c)}{\leq} \frac{\tau+1}{2}\|\bar{\mv{u}}_l^H\tilde{\mv{H}}\|_1,
\end{align}
where ($c$) holds since $\tau\geq 0$. Therefore, $\bar{\mv{u}}_l$ satisfies the constraint in (\ref{P2f1}). The proof of Proposition \ref{prop_BFfeas} is thus completed.
\section{Proof of Lemma \ref{lemma_P2feas}}\label{proof_lemma_P2feas}
We first show that to prove Lemma \ref{lemma_P2feas}, it suffices to show that the optimal ${\mv{u}}_f$ to Problem (P2-FS) denoted by ${\mv{u}}_f^\star$ satisfies $\|{\mv{u}}_f^{\star H}\tilde{\mv{H}}\|_1>0$. Specifically, if $\|{\mv{u}}_f^{\star H}\tilde{\mv{H}}\|_1>0$, ${\mv{u}}_f^\star$ can be shown to be a feasible solution to Problem (P2-F) according to the proof of Proposition 1, by noting that it satisfies the constraint in (\ref{P2FS1}); moreover, since $\|{\mv{u}}_f^\star\|_2\leq 1$ holds, it is a feasible solution to Problem (P2).

Then, we prove $\|{\mv{u}}_f^{\star H}\tilde{\mv{H}}\|_1>0$ by contradiction. Suppose, on the contrary, that $\|{\mv{u}}_f^{\star H}\tilde{\mv{H}}\|_1=0$ holds. By noting that $\|{\mv{u}}_f^H\tilde{\mv{H}}\|_1\geq \mathfrak{Re}\left\{{\mv{u}}_f^H\sum_{i=1}^{M_t}\tilde{\mv{h}}_i^H\right\}$ holds for any ${\mv{u}}_f$, it follows that the optimal value of Problem (P2-FS) for any given $j$ is zero, which implies $\mathrm{Null}\left(2\tilde{\mv{h}}_j^H-\sum_{i=1}^{M_t}\tilde{\mv{h}}_i^H\right)\subseteq \mathrm{Null}\left(\sum_{i=1}^{M_t}\tilde{\mv{h}}_i^H\right),\ \forall j$. However, also note that this is only true if there exist $\beta_j$'s such that $\tilde{\mv{h}}_j^H=\beta_j\sum_{i=1}^{M_t}\tilde{\mv{h}}_i^H,\ \forall j$ holds, i.e., $\mathrm{rank}(\tilde{\mv{H}})= 1$. This contradicts our assumption of $\mathrm{rank}(\tilde{\mv{H}})\geq 2$. The proof of Lemma \ref{lemma_P2feas} is thus completed.
\section{Proof of Proposition \ref{prop_P2equi}} \label{proof_prop_P2equi}
First, given any feasible solution ${\mv{u}}$ to Problem (P2), it follows from (\ref{equiobj}) and (\ref{optp}) that $({\mv{u}},{\mv{p}}^\star({\mv{u}}))$ is feasible for Problem (P3) and achieves the same objective value as that of Problem (P2), thus the optimal value of Problem (P3) is no smaller than that of Problem (P2). On the other hand, it follows from (\ref{optp}) that the objective value of Problem (P3) with any feasible solution $({\mv{u}},{\mv{p}})$ is always no larger than that with the solution $({\mv{u}},{\mv{p}}^\star({\mv{u}}))$. Moreover, note that based on Lemma \ref{lemma_P2feas} and (\ref{optp}), $({\mv{u}}_f^\star,{\mv{p}}^\star({\mv{u}}_f^\star))$ is a feasible solution to Problem (P3) and yields $\mathfrak{Re}\left\{{\mv{u}}_f^{\star H}\tilde{\mv{H}}{\mv{p}}^\star({\mv{u}}_f^\star)\right\}>0$, where ${\mv{u}}_f^{\star}$ denotes the optimal ${\mv{u}}_f$ to Problem (P2-FS). It then follows that the optimal solution to Problem (P3) denoted by $({\mv{u}}^\star,{\mv{p}}^\star({\mv{u}}^\star))$ satisfies $\|{\mv{u}}^{\star H}\tilde{\mv{H}}\|_1=\mathfrak{Re}\left\{{\mv{u}}^{\star H}\tilde{\mv{H}}{\mv{p}}^\star({\mv{u}^\star})\right\}>0$, thus ${\mv{u}}^\star$ is feasible for Problem (P2) and achieves the same objective value as that of Problem (P3) with the optimal solution. Hence, the optimal value of Problem (P2) is no smaller than that of Problem (P3). Therefore, Problems (P2) and (P3) have the same optimal value. This thus completes the proof of Proposition \ref{prop_P2equi}.
\section{Proof of Proposition \ref{prop_P5eqv}}\label{proof_prop_P5eqv}
First, given any feasible solution ${\mv{u}}_k$ to Problem (P4), it can be shown that $\left({\mv{u}}_k\big/\sqrt{\|{\mv{u}}_k^H\tilde{\mv{H}}_{[-k]}\|_2^2+\tilde{\sigma}^2},\right.$\\$\left.1\big/\sqrt{\|{\mv{u}}_k^H\tilde{\mv{H}}_{[-k]}\|_2^2+\tilde{\sigma}^2}\right)$ is feasible for Problem (P5) and achieves the same objective value as that of Problem (P4). On the other hand, given any feasible solution $({\mv{u}}_k,t_k)$ to Problem (P5), it can be shown that ${\mv{u}}_k/\sqrt{t_k}$ is a feasible solution to Problem (P4) and achieves the same objective value as that of Problem (P5). Therefore, Problem (P4) and Problem (P5) have the same optimal value. The proof of Proposition \ref{prop_P5eqv} is thus completed.
\section{Proof of Proposition \ref{prop_SMfeas}}\label{proof_prop_SMfeas}
We prove Proposition \ref{prop_SMfeas} by showing that for any transmit antenna grouping, $M_r\geq \frac{(K-1)M_t}{K}+2$ and $\frac{M_t}{K}\geq 2$ yield $\mathrm{rank}({\mv{H}}_k)\geq 2,\ \forall k$, thus guaranteeing the feasibility of Problem (P7) for all $k$'s according to Proposition \ref{prop_BFfeas}.

First, note that $\mathrm{rank}(\tilde{\mv{V}}_k)=M_r-\mathrm{rank}(\tilde{\mv{H}}_{[-k]})=M_r-\frac{(K-1)M_t}{K}$. Then, by noting that $\tilde{\mv{V}}_k^H\tilde{\mv{H}}_{[-k]}={\mv{0}}$, it can be shown that $\mathrm{rank}({\mv{H}}_k)=\mathrm{rank}\left(\tilde{\mv{V}}_k^H\tilde{\mv{H}}_k\right)=\mathrm{rank}\left(\tilde{\mv{V}}_k^H\tilde{\mv{H}}\right)\geq \mathrm{rank}\left(\tilde{\mv{V}}_k^H\right)+\mathrm{rank}(\tilde{\mv{H}})-M_r=\min\{M_r,M_t\}-\frac{(K-1)M_t}{K}$. Hence, $M_r\geq \frac{(K-1)M_t}{K}+2$ and $\frac{M_t}{K}\geq 2$ suffice to ensure $\mathrm{rank}({\mv{H}}_k)\geq 2$. The proof of Proposition \ref{prop_SMfeas} is thus completed.
\bibliographystyle{IEEEtran}
\bibliography{CEMIMO}

\begin{thebibliography}{10}
\providecommand{\url}[1]{#1}
\csname url@samestyle\endcsname
\providecommand{\newblock}{\relax}
\providecommand{\bibinfo}[2]{#2}
\providecommand{\BIBentrySTDinterwordspacing}{\spaceskip=0pt\relax}
\providecommand{\BIBentryALTinterwordstretchfactor}{4}
\providecommand{\BIBentryALTinterwordspacing}{\spaceskip=\fontdimen2\font plus
\BIBentryALTinterwordstretchfactor\fontdimen3\font minus
  \fontdimen4\font\relax}
\providecommand{\BIBforeignlanguage}[2]{{%
\expandafter\ifx\csname l@#1\endcsname\relax
\typeout{** WARNING: IEEEtran.bst: No hyphenation pattern has been}%
\typeout{** loaded for the language `#1'. Using the pattern for}%
\typeout{** the default language instead.}%
\else
\language=\csname l@#1\endcsname
\fi
#2}}
\providecommand{\BIBdecl}{\relax}
\BIBdecl

\bibitem{ICC16}
S.~Zhang, R.~Zhang, and T.~J. Lim, ``Receive beamforming optimization for
  {MIMO} system with constant envelope precoding,'' in \emph{Proc. IEEE Int.
  Conf. Commun. (ICC)}, May 2016, pp. 1--6.

\bibitem{SUCE}
S.~K. Mohammed and E.~G. Larsson, ``Single-user beamforming in large-scale
  \protect{MISO} systems with per-antenna constant-envelope constraints: the
  doughnut channel,'' \emph{IEEE Trans. Wireless Commun.}, vol.~11, no.~11, pp.
  3992--4005, Nov. 2012.

\bibitem{JSTSPCE}
J.~Pan and W.-K. Ma, ``Constant envelope precoding for single-user large-scale
  \protect{MISO} channels: efficient precoding and optimal designs,''
  \emph{IEEE J. Sel. Topics Signal Process.}, vol.~8, no.~5, pp. 982--995, Oct.
  2014.

\bibitem{CEadaptive}
S.~Zhang, R.~Zhang, and T.~J. Lim, ``Constant envelope precoding with adaptive
  receiver constellation in {MISO} fading channel,'' \emph{IEEE Trans. Wireless
  Commun.}, to appear. [Online]. Available:
  \url{http://arxiv.org/abs/1503.09178}.

\bibitem{CEMulticast}
------, ``{MISO} multicasting with constant envelope precoding,'' \emph{IEEE
  Wireless Commun. Lett.}, to appear.

\bibitem{MUCE}
S.~K. Mohammed and E.~G. Larsson, ``Per-antenna constant envelope precoding for
  large multi-user \protect{MIMO} systems,'' \emph{IEEE Trans. Commun.},
  vol.~61, no.~3, pp. 1059--1071, Mar. 2013.

\bibitem{improved}
J.-C. Chen, C.-K. Wen, and K.-K. Wong, ``Improved constant envelope multiuser
  precoding for massive {MIMO} systems,'' \emph{IEEE Commun. Lett.}, vol.~18,
  no.~8, pp. 1311--1314, Aug. 2014.

\bibitem{freqCE}
S.~K. Mohammed and E.~G. Larsson, ``Constant-envelope multi-user precoding for
  frequency-selective massive \protect{MIMO} systems,'' \emph{IEEE Wireless
  Commun. Lett.}, vol.~2, no.~5, pp. 547--550, Oct. 2013.

\bibitem{phaseconstraint}
S.~Mukherjee and S.~K. Mohammed, ``Constant envelope precoding with
  time-variation constraint on the transmitted phase angles,'' \emph{IEEE
  Wireless Commun. Lett.}, vol.~4, no.~2, pp. 221--224, Apr. 2015.

\bibitem{RFPA}
S.~C. Cripps, \emph{\protect{RF} \protect{Power Amplifiers for Wireless
  Communications}}.\hskip 1em plus 0.5em minus 0.4em\relax Artech
  \protect{Publishing House}, 1999.

\bibitem{phaseshifter}
UKRF, ``Analog and digital phase shifters,'' [Online]. Available at:
  \url{http://www.ukrf.com/controlpanel/shoppics/pdfs/MiteqAnalogueDigitalPhaseShifters.pdf}.

\bibitem{WSA14}
A.~S. Mohammed, R.~R. Muller, and G.~Fischer, ``A novel single-{RF} transmitter
  for massive {MIMO},'' in \emph{Proc. Int. ITG Workshop Smart Antenna (WSA)},
  Mar. 2014, pp. 1--8.

\bibitem{linearprecoding}
A.~Wiesel, Y.~C. Eldar, and S.~Shamai~(Shitz), ``Linear precoding via conic
  optimization for fixed {MIMO} receivers,'' \emph{IEEE Trans. Signal
  Process.}, vol.~54, no.~1, pp. 161--176, Jan. 2006.

\bibitem{dualityPAPC}
W.~Yu and T.~Lan, ``Transmitter optimization for the multi-antenna downlink
  with per-antenna power constraints,'' \emph{IEEE Trans. Signal Process.},
  vol.~55, no.~6, pp. 2646--2660, Jun. 2007.

\bibitem{LTCC}
L.~Zhang, R.~Zhang, Y.-C. Liang, Y.~Xin, and H.~V. Poor, ``On
  \protect{Gaussian} \protect{MIMO} \protect{BC-MAC} duality with multiple
  transmit covariance constraints,'' \emph{IEEE Trans. Inf. Theory}, vol.~58,
  no.~4, pp. 2064--2078, Apr. 2012.

\bibitem{multicellBD}
R.~Zhang, ``Cooperative multi-cell block diagonalization with per-base-station
  power constraints,'' \emph{IEEE J. Sel. Areas Commun.}, vol.~28, no.~9, pp.
  1435--1445, Dec. 2010.

\bibitem{massivePAPC}
S.~Zhang, R.~Zhang, and T.~J. Lim, ``Massive \protect{MIMO} with per-antenna
  power constraint,'' in \emph{Proc. IEEE Global Conf. on Signal and Inf.
  Process. (GlobalSIP)}, Dec. 2014, pp. 642--646.

\bibitem{digicom}
J.~G. Proakis and M.~Salehi, \emph{Digital Communications}.\hskip 1em plus
  0.5em minus 0.4em\relax McGraw-Hill, 2008.

\bibitem{Huang06}
S.~Zhang and Y.~Huang, ``Complex quadratic optimization and semidefinite
  programming,'' \emph{SIAM J. Optim.}, vol.~16, no.~3, pp. 871--890, 2006.

\bibitem{Luo10}
Z.-Q. Luo, W.-K. Ma, A.-C. So, Y.~Ye, and S.~Zhang, ``Semidefinite relaxation
  of quadratic optimization problems,'' \emph{IEEE Signal Process. Mag.},
  vol.~27, no.~3, pp. 20--34, May 2010.

\bibitem{Stoica14}
M.~Soltanalian and P.~Stoica, ``Designing unimodular codes via quadratic
  optimization,'' \emph{IEEE Trans. Signal Process.}, vol.~62, no.~5, pp.
  1221--1234, Mar. 2014.

\bibitem{cvx}
M.~Grant and S.~Boyd, ``{CVX: Matlab software for disciplined convex
  programming},'' version 2.1. [Online] Available: \url{http://cvxr.com/cvx/},
  Jun. 2015.

\bibitem{Charnes}
A.~Charnes and W.~W. Cooper, ``Programming with linear fractional functions,''
  \emph{Naval Res. Logist. Quarter.}, vol.~9, pp. 181--186, Dec. 1962.

\end{thebibliography}
\end{document}